\pdfoutput=1
\RequirePackage{ifpdf}
\ifpdf % We~are running pdfTeX in pdf mode
\documentclass[pdftex]{sigma}
\else
\documentclass{sigma}
\fi

\newcommand{\inversion}[1]{\!\overleftrightarrow{\,#1\,}\!}

\newcommand{\reverse}{{\,\overleftrightarrow{\!(\cdots)\!}\,}}

\usepackage{tikz}

\newcommand{\crosssquare}{{\begin{tikzpicture}[scale=.15]
			\draw (0,0) -- (1,1) -- (1,0) -- (0,1) -- (0,0) -- (1,0) -- (1,1) -- (0,1) -- cycle;
\end{tikzpicture}}}

\newcommand{\cross}{{\begin{tikzpicture}[scale=.15]
			\draw (0,0) -- (1,1);
			\draw (1,0) -- (0,1);
\end{tikzpicture}}}

\newcommand{\triang}{{\begin{tikzpicture}[scale=.15]
			\draw (1,0) -- (0,0) -- (0,1) -- cycle;
\end{tikzpicture}}}

\newcommand{\trileg}{{\begin{tikzpicture}[scale=.15]
			\draw (1,1) -- (0,0);
			\draw (1,0) -- (0,0) -- (0,1);
\end{tikzpicture}}}

\newcommand{\cyclic}{{\circlearrowleft}_{ijk}}

\newcommand{\pdv}[2]{\frac{\partial {#1}}{\partial {#2}}}

\newcommand{\C}{\mathbb{C}}

\newcommand{\Z}{\mathbb{Z}}

\newcommand{\cL}{\mathcal{L}}

\renewcommand{\vec}{\boldsymbol}

{ \theoremstyle{definition}
	\newtheorem{Definition}{Definition}[section]
	\newtheorem{Example}[Definition]{Example}

}
\newtheorem{Lemma}[Definition]{Lemma}
\newtheorem{Theorem}[Definition]{Theorem}
\newtheorem{Proposition}[Definition]{Proposition}

\newtheorem{Assumption}[Definition]{Assumption}

\numberwithin{equation}{section}

\begin{document}
%\allowdisplaybreaks

\newcommand{\arXivNumber}{2403.16845}

\renewcommand{\PaperNumber}{059}

\FirstPageHeading

\ShortArticleName{Discrete Lagrangian Multiforms for ABS Equations II:}

\ArticleName{Discrete Lagrangian Multiforms for ABS Equations~II:\\ Tetrahedron and Octahedron Equations}

\Author{Jacob J. RICHARDSON~$^{\rm a}$ and Mats VERMEEREN~$^{\rm b}$}

\AuthorNameForHeading{J.J.~Richardson and M.~Vermeeren}

\Address{$^{\rm a)}$~School of Mathematics, University of Leeds, Leeds, LS2 9JT, UK}
\EmailD{\href{mailto:jacobjoseph.gu@gmail.com}{jacobjoseph.gu@gmail.com}}

\Address{$^{\rm b)}$~Department of Mathematical Sciences, Loughborough University, \\
\hphantom{$^{\rm b)}$}~Loughborough, LE11 3TU, UK}
\EmailD{\href{mailto:m.vermeeren@lboro.ac.uk}{m.vermeeren@lboro.ac.uk}}
	
\ArticleDates{Received March 26, 2024, in final form July 02, 2025; Published online July 18, 2025}

\Abstract{We present four types of discrete Lagrangian 2-form associated to the integrable quad equations of the ABS list. These include the triangle Lagrangian that has traditionally been used in the Lagrangian multiform description of ABS equations, the trident Lagrangian that was central to Part I of this paper, and two Lagrangians that have not been studied in the multiform setting. Two of the Lagrangian 2-forms have the quad equations, or a system equivalent to the quad equations, as their Euler--Lagrange equations, and one produces the tetrahedron equations. This is in contrast to the triangle Lagrangian 2-form, which produces equations that are weaker than the quad equations (they are equivalent to two octahedron equations). We use relations between the Lagrangian 2-forms to prove that the system of quad equations is equivalent to the combined system of tetrahedron and octahedron equations. Furthermore, for each of the Lagrangian 2-forms, we study the double zero property of the exterior derivative. In particular, this gives a possible variational interpretation to the octahedron equations.}

	
\Classification{39A36; 37J70; 37J06}

\section{Introduction}

In Part I of this work~\cite{richardson2025discrete1}, we revisited the discrete Lagrangian multiforms for the integrable quad equations of the ABS list. We showed that two central properties of the theory are more subtle than usually acknowledged:
\begin{itemize}	\itemsep=0pt
 \item[(1)] the equivalence between quad equations and their three-leg forms,
 \item[(2)] the closure relation of the Lagrangian multiforms.
\end{itemize}
By introducing additional, integer-valued, fields to the action, we managed to put property~(1) on sound footing and recovered a slightly weaker version of property~(2). We gave counterexamples to the usual formulations of these properties.
In Part~I, we focused our attention to what we call the \emph{trident Lagrangian}. We chose this over the established choice in discrete Lagrangian multiform theory (which we call the \emph{triangle Lagrangian} and was introduced in~\cite{lobb2009lagrangian}, see also~\mbox{\cite{bobenko2010lagrangian, boll2014integrability, boll2016integrability, lobb2018variational, xenitidis2011lagrangian}}), because the trident Lagrangian produces Euler--Lagrange equations that are equivalent to the multi-affine quad equations, whereas the triangle Lagrangian produces a~weaker set of equations.

The first aim of Part~II is to compare these choices of Lagrangian, as well as two further Lagrangians. We will show how the Euler--Lagrange equations of these different Lagrangian multiforms relate to various sets of equations that have been associated to integrable quad equations: tetrahedron equations and octahedron relations~\cite{boll2014integrability, boll2016integrability}.

The second aim of Part II involves an algebraic interpretation of the closure relation that has recently been emphasised: we can write the exterior derivative of the Lagrangian form as (a~sum~of) product(s) of Euler--Lagrange expressions. In other words, the exterior derivative is not just zero on solutions of the Euler--Lagrange equations \big(or possibly a multiple of $4 \pi^2$\big), but attains a~\emph{double zero} on this set of equations.
In the continuous (and semi-discrete) case, the double zero property has been used implicitly in~\cite{sleigh2020variational, suris2016variational} and discussed explicitly in~\mbox{\cite{caudrelier2023lagrangiana, petrera2021variational, sleigh2022semi}}. In the discrete setting, a double zero expansion for the lattice Boussinesq equation was recently obtained in~\cite{nijhoff2024lagrangian}. In the present work, we give double zero expansions for the Lagrangian multiforms of the ABS list.
These double zero expansions extend the variational interpretation of the quad and tetrahedron equations and provide the first known variational interpretation of the octahedron equations.

The structure of this paper is as follows. In Section~\ref{sec-abs}, we review the properties of the ABS equations that are relevant to this work. We start Section~\ref{sec-multiforms} with a review of Lagrangian multiform theory, including the extension of the action by integer fields that was developed in Part~I. In Sections~\ref{sec-cross}--\ref{sec-triangle}, we present four types of Lagrangian multiform. For two of these, the variational principle produces exactly the set of quad equations. One of the Lagrangian multiforms produces only the tetrahedron equations. The final one is the well-known Lagrangian multiform on a triangular stencil. In Section~\ref{sec-octa}, we show that the relations between quad equations, tetrahedron equations, and octahedron equations follow from relations between the different Lagrangian 2-forms. In Section~\ref{sec-double0}, we show that the exterior derivatives of all four 2-forms admit a double zero expansion (in terms of quad, tetrahedron, or octahedron polynomials). We give a general construction of these double zero expansions and explicit expressions for the quad equation H1.

\section{Quad equations}\label{sec-abs}

We consider the multi-affine and multidimensionally consistent quad equations from the ABS list~\cite{adler2003classification}. The ABS list consists of
\begin{itemize}	\itemsep=0pt
 \item equations of Hirota-type H1, H2, H3$_{\delta}$,
 \item more symmetric equations Q1$_{\delta}$, Q2, Q3$_{\delta}$, Q4,
 \item equations A1$_{\delta}$, A2, related to Q1$_{\delta}$, Q3 by non-autonomous transformations.
\end{itemize}
The subscript in $Q1_{\delta}$ etc.\@ indicates dependence on a parameter $\delta$. We often distinguish between zero and nonzero values of the parameter by writing $Q1_{\delta=0}$ or $Q1_{\delta\neq0}$.
Details of these equations can be found, for example, in~\cite{adler2003classification, bobenko2010lagrangian, hietarinta2009soliton, nijhoff2009soliton}, or in Part I~\cite{richardson2025discrete1}.

\subsection{Three-leg forms}

Given a field $u\colon \Z^N \to \C$ and a reference point $\vec n \in \Z^N$, we write $u = u(\vec n)$, $u_i = u(\vec n + \vec e_i)$, $u_{ij} = u(\vec n + \vec e_i + \vec e_j)$, etc., where $\vec e_i$ denotes the unit vector in the $i$-th direction.
Each equation
\[ Q_{ij} := Q(u,u_i,u_j,u_{ij}, \alpha_i, \alpha_j) = 0\]
of the ABS list is related to a \emph{three-leg form}. For some of the equations (H1, A1$_{\delta=0}$ and Q1$_{\delta=0}$), there exist functions $\psi$ and $\phi$ such that the multi-affine equation is equivalent to
\begin{equation}
 \label{three-leg-additive-0}
 \phi(u,u_{ij}, \alpha_i-\alpha_j) = \psi(u,u_i,\alpha_i) - \psi(u,u_j,\alpha_j) .
\end{equation}
The other ABS equations have three-leg form with multiplicative structure: there exist functions~$\Psi$ and~$\Phi$ such that the multi-affine equation is equivalent to
\[ \Phi(u,u_{ij}, \alpha_i-\alpha_j) = \frac{ \Psi(u,u_i,\alpha_i) }{ \Psi(u,u_j,\alpha_j) } . \]
Setting $\psi = \log(\Psi)$ and $\phi = \log(\phi)$, we find that the multi-affine equation (for H2, H3, A1$_{\delta\neq0}$, A2, Q1$_{\delta\neq0}$, Q2, Q3, Q4) is equivalent to
\begin{equation}
 \label{three-leg-additive-2pii}
 \phi(u,u_{ij}, \alpha_i-\alpha_j) = \psi(u,u_i,\alpha_i) - \psi(u,u_j,\alpha_j) + 2 \Theta \pi {\rm i}, \qquad \Theta \in \Z .
\end{equation}
For all equations of the ABS list, regardless of whether the three leg form is of type~\eqref{three-leg-additive-0} or~\eqref{three-leg-additive-2pii}, $\phi$ is an odd function of its last entry, i.e., $\phi(u,u_{ij}, \alpha_i-\alpha_j) = - \phi(u,u_{ij}, \alpha_j-\alpha_i)$.

We denote the additive three-leg expression by
\begin{equation*}
	\mathcal{Q}_{ij}^{(u)} := \psi(u,u_i,\alpha_i) - \psi(u,u_j,\alpha_j) - \phi(u,u_{ij}, \alpha_i-\alpha_j) ,
	%\label{quad-u}
\end{equation*}
where the superscript indicates that the three legs meet at the vertex $u$. We can also consider three-leg forms based at the other three vertices of the square:
\begin{align}
	&\mathcal{Q}_{ij}^{(u_i)}:= \psi(u_i,u_{ij},\alpha_j) - \psi(u_i,u,\alpha_i) - \phi(u_i,u_j, \alpha_j-\alpha_i) , \label{quad-ui}\\
	&\mathcal{Q}_{ij}^{(u_{ij})}:= \psi(u_{ij},u_j,\alpha_i) - \psi(u_{ij},u_i,\alpha_j) - \phi(u_{ij},u, \alpha_i-\alpha_j) , \label{quad-uij}\\
	&\mathcal{Q}_{ij}^{(u_j)}:= \psi(u_j,u,\alpha_j) - \psi(u_j,u_{ij},\alpha_i) - \phi(u_j,u_i, \alpha_j-\alpha_i) . \label{quad-uj}
\end{align}
In summary, we have the following.

\begin{Proposition} \label{prop-three-leg-quad}
 For ${\rm H}1$, ${\rm A}1_{\delta=0}$ and ${\rm Q}1_{\delta=0}$, there holds
 \[ Q_{ij} = 0 \iff \mathcal{Q}_{ij}^{(u)} = 0 \iff \mathcal{Q}_{ij}^{(u_i)} = 0 \iff \mathcal{Q}_{ij}^{(u_j)} = 0 \iff \mathcal{Q}_{ij}^{(u_{ij})} = 0 .\]
 For ${\rm H}2$, ${\rm H}3$, ${\rm A}1_{\delta\neq0}$, ${\rm A}2$, ${\rm Q}1_{\delta\neq0}$, ${\rm Q}2$, ${\rm Q}3$, there holds
 \begin{align*} Q_{ij} = 0 &\iff \mathcal{Q}_{ij}^{(u)} \equiv 0 \mod 2 \pi {\rm i} \iff \mathcal{Q}_{ij}^{(u_i)} \equiv 0 \mod 2 \pi {\rm i} \\
 &\iff \mathcal{Q}_{ij}^{(u_j)} \equiv 0 \mod 2 \pi {\rm i} \iff
 \mathcal{Q}_{ij}^{(u_{ij})} \equiv 0 \mod 2 \pi {\rm i} .
 \end{align*}
\end{Proposition}

The three-leg forms of quad equations on adjacent faces of the cube combine to form an equation on a tetrahedral stencil. We can base the three-leg form of the tetrahedron equation at any of its vertices, as illustrated in Figure~\ref{fig-tetra},
\begin{align}
 &\mathcal{T}^{(u)}:= \phi(u,u_{ij},\alpha_i-\alpha_j) + \phi(u,u_{jk},\alpha_j-\alpha_k) + \phi(u,u_{ki},\alpha_k-\alpha_i) , \label{tetra-u}\\
	&\mathcal{T}^{(u_{ij})}:= \phi(u_{ij},u,\alpha_i-\alpha_j) + \phi(u_{ij},u_{ki},\alpha_j-\alpha_k) + \phi(u_{ij},u_{jk}, \alpha_k-\alpha_i) , \label{tetra-uij}\\
	&\mathcal{T}^{(u_{ki})}:= \phi(u_{ki},u_{jk},\alpha_i-\alpha_j) + \phi(u_{ki},u_{ij},\alpha_j-\alpha_k) + \phi(u_{ki},u, \alpha_k-\alpha_i) , \label{tetra-uki}\\
	&\mathcal{T}^{(u_{jk})}:= \phi(u_{jk},u_{ki},\alpha_i-\alpha_j) + \phi(u_{jk},u,\alpha_j-\alpha_k) + \phi(u_{jk},u_{ij}, \alpha_k-\alpha_i) . \label{tetra-ujk}
\end{align}

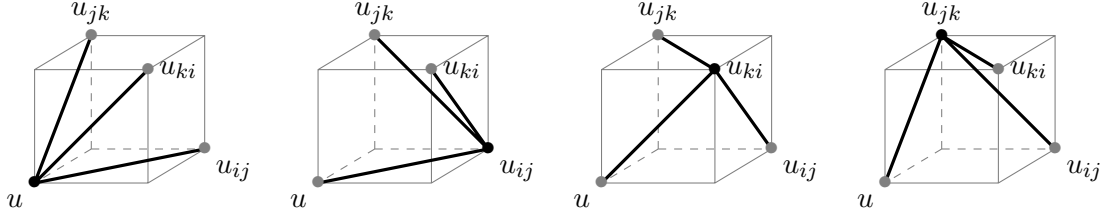
\begin{figure}[t]
\centering
\begin{tikzpicture}[scale=1.5]
	\begin{scope}[y={(5mm, 3mm)}, z={(0cm,1cm)}]
 	\draw[gray, every edge/.append style={dashed}]
		(1,0,0) -- (0,0,0) -- (0,0,1) -- (0,1,1) -- (1,1,1) -- (1,1,0) -- cycle -- (1,0,1) -- (0,0,1)
		(1,0,1) -- (1,1,1)
		(0,1,0) edge (0,0,0) edge (0,1,1) edge (1,1,0);
		\draw[very thick, every edge/.append style={dashed}]
 (0,0,0) node[below left] {$u$} -- (0,1,1) node[above] {$u_{jk}$}
		(0,0,0) -- (1,0,1) node[right] {$u_{ki}$}
		(0,0,0) -- (1,1,0) node[below right] {$u_{ij}$};
 \node at (0,0,0) {$\bullet$};
 \node[gray] at (1,1,0) {$\bullet$};
 \node[gray] at (0,1,1) {$\bullet$};
 \node[gray] at (1,0,1) {$\bullet$};
	\end{scope}
 \begin{scope}[shift=({2.5,0}),scale=1, y={(5mm, 3mm)}, z={(0cm,1cm)}]
 	\draw[gray, every edge/.append style={dashed}]
		(1,0,0) -- (0,0,0) -- (0,0,1) -- (0,1,1) -- (1,1,1) -- (1,1,0) -- cycle -- (1,0,1) -- (0,0,1)
		(1,0,1) -- (1,1,1)
		(0,1,0) edge (0,0,0) edge (0,1,1) edge (1,1,0);
		\draw[very thick, every edge/.append style={dashed}]
 (1,1,0) node[below right] {$u_{ij}$} -- (0,0,0) node[below left] {$u$}
 (1,1,0) -- (0,1,1) node[above] {$u_{jk}$}
		(1,1,0) -- (1,0,1) node[right] {$u_{ki}$};
 \node[gray] at (0,0,0) {$\bullet$};
 \node at (1,1,0) {$\bullet$};
 \node[gray] at (0,1,1) {$\bullet$};
 \node[gray] at (1,0,1) {$\bullet$};
	\end{scope}	
 \begin{scope}[shift=({5,0}),scale=1, y={(5mm, 3mm)}, z={(0cm,1cm)}]
 	\draw[gray, every edge/.append style={dashed}]
		(1,0,0) -- (0,0,0) -- (0,0,1) -- (0,1,1) -- (1,1,1) -- (1,1,0) -- cycle -- (1,0,1) -- (0,0,1)
		(1,0,1) -- (1,1,1)
		(0,1,0) edge (0,0,0) edge (0,1,1) edge (1,1,0);
		\draw[very thick, every edge/.append style={dashed}]
 (1,0,1) node[right] {$u_{ki}$} -- (1,1,0) node[below right] {$u_{ij}$}
 (1,0,1) -- (0,0,0) node[below left] {$u$}
 (1,0,1) -- (0,1,1) node[above] {$u_{jk}$};
 \node[gray] at (0,0,0) {$\bullet$};
 \node[gray] at (1,1,0) {$\bullet$};
 \node[gray] at (0,1,1) {$\bullet$};
 \node at (1,0,1) {$\bullet$};
	\end{scope}	
 \begin{scope}[shift=({7.5,0}),scale=1, y={(5mm, 3mm)}, z={(0cm,1cm)}]
 	\draw[gray, every edge/.append style={dashed}]
		(1,0,0) -- (0,0,0) -- (0,0,1) -- (0,1,1) -- (1,1,1) -- (1,1,0) -- cycle -- (1,0,1) -- (0,0,1)
		(1,0,1) -- (1,1,1)
		(0,1,0) edge (0,0,0) edge (0,1,1) edge (1,1,0);
		\draw[very thick, every edge/.append style={dashed}]
 (0,1,1) node[above] {$u_{jk}$} -- (1,0,1) node[right] {$u_{ki}$}
 (0,1,1) -- (1,1,0) node[below right] {$u_{ij}$}
 (0,1,1) -- (0,0,0) node[below left] {$u$};
 \node[gray] at (0,0,0) {$\bullet$};
 \node[gray] at (1,1,0) {$\bullet$};
 \node at (0,1,1) {$\bullet$};
 \node[gray] at (1,0,1) {$\bullet$};
	\end{scope}	
\end{tikzpicture}
\caption{Four three-leg forms of a tetrahedron equation.}
\label{fig-tetra}
\end{figure}

The ABS equations of type Q have the property that their short and long leg functions are the same, $\phi=\psi$. Each equation of type H and A shares its long leg function $\phi$ with an equation of type Q, but has a different short leg function $\psi$.
Thus, the tetrahedron equation of any member of the ABS list takes the form of a quad equation of type Q, in their respective three-leg forms. Similarly, the multi-affine tetrahedron equation $T=0$ can be described as a quad polynomial of type Q, evaluated on a tetrahedron stencil:
\[ T(u,u_{ij},u_{jk},u_{ki},\alpha_i,\alpha_j,\alpha_k)
 = Q^{(\text{type Q})}(u,u_{ij},u_{jk},u_{ki},\alpha_i-\alpha_j,-\alpha_j+\alpha_k) . \]
A second tetrahedron equation is found from this by point inversion,
\[ T(u_{ijk},u_{k},u_{i},u_{j},\alpha_i,\alpha_j,\alpha_k)
 = Q^{(\text{type Q})}(u_{ijk},u_{k},u_{i},u_{j},\alpha_i-\alpha_j,-\alpha_j+\alpha_k) . \]
Its three-leg forms are
\begin{align}
 &\mathcal{T}^{(u_{ijk})}:= \phi(u_{ijk},u_{k},\alpha_i-\alpha_j) + \phi(u_{ijk},u_{i},\alpha_j-\alpha_k) + \phi(u_{ijk},u_{j},\alpha_k-\alpha_i) ,\label{tetra-uijk}
 \\
	&\mathcal{T}^{(u_{k})}:= \phi(u_{k},u_{ijk},\alpha_i-\alpha_j) + \phi(u_{k},u_{j},\alpha_j-\alpha_k) + \phi(u_{k},u_{i}, \alpha_k-\alpha_i) \label{tetra-uk} ,\\
	&\mathcal{T}^{(u_{j})}:= \phi(u_{j},u_{i},\alpha_i-\alpha_j) + \phi(u_{j},u_{k},\alpha_j-\alpha_k) + \phi(u_{j},u_{ijk}, \alpha_k-\alpha_i) \label{tetra-uj} ,\\
	&\mathcal{T}^{(u_{i})}:= \phi(u_{i},u_{j},\alpha_i-\alpha_j) + \phi(u_{i},u_{ijk},\alpha_j-\alpha_k) + \phi(u_{i},u_{k}, \alpha_k-\alpha_i) \label{tetra-ui} .
\end{align}

The equivalences between the tetrahedron equations in multi-affine and three-leg forms are as follows.
\begin{Proposition}
 %\label{prop-three-leg-tetra}
 For ${\rm H}1$, ${\rm A}1_{\delta=0}$ and ${\rm Q}1_{\delta=0}$, the multi-affine equation
 \[ T(u,u_{ij},u_{jk},u_{ki},\alpha_i,\alpha_j,\alpha_k) = 0\] is equivalent to each of the following:
 \[ \mathcal{T}^{(u)} = 0 \iff \mathcal{T}^{(u_{ij})} = 0 \iff \mathcal{T}^{(u_{jk})} = 0 \iff \mathcal{T}^{(u_{ki})} = 0 .\]

 For ${\rm H}2$, ${\rm H}3$, ${\rm A}1_{\delta\neq0}$, ${\rm A}2$, ${\rm Q}1_{\delta\neq0}$, ${\rm Q}2$, ${\rm Q}3$,
 the multi-affine equation
 \[ T(u,u_{ij},u_{jk},u_{ki},\alpha_i,\alpha_j,\alpha_k) = 0\] is equivalent to each of the following:
 \begin{gather*}
 \mathcal{T}^{(u)} \equiv 0 \mod 2 \pi {\rm i} \iff
 \mathcal{T}^{(u_{ij})} \equiv 0 \mod 2 \pi {\rm i} \iff \\
 \mathcal{T}^{(u_{jk})} \equiv 0 \mod 2 \pi {\rm i} \iff
 \mathcal{T}^{(u_{ki})} \equiv 0 \mod 2 \pi {\rm i} .
 \end{gather*}
Analogous equivalences relate $T(u_{ijk},u_{k},u_{i},u_{j},\alpha_i,\alpha_j,\alpha_k) = 0$ to $\mathcal{T}^{(u_{ijk})}$, $\mathcal{T}^{(u_{k})}$, $\mathcal{T}^{(u_{i})}$, and $\mathcal{T}^{(u_{j})}$.
\end{Proposition}

\section{Lagrangian multiforms}
\label{sec-multiforms}

\subsection{General theory}

The central object of the Lagrangian multiform description of quad equations is a discrete two-form $\cL(u,u_i,u_j,u_{ij},\alpha_i,\alpha_j)$, i.e., a function that is skew-symmetric under the swap of indices~${i \leftrightarrow j}$. We require that, for every choice of 2-dimensional discrete surface in the lattice, the corresponding action sum is critical with respect to variations of the field $u\colon \Z^N \to \C$. Equivalently, we require that the action sum over every elementary cube,{\samepage
\begin{align*}
\begin{split}
 S &= \cL(u_k,u_{ki},u_{jk}, u_{ijk},\alpha_i,\alpha_j) + \cL(u_i,u_{ij},u_{ki}, u_{ijk},\alpha_j,\alpha_k) + \cL(u_j,u_{jk},u_{ij}, u_{ijk},\alpha_k,\alpha_i)
 \\ &\quad{} - \cL(u,u_i,u_j,u_{ij},\alpha_i,\alpha_j) - \cL(u,u_j,u_k,u_{jk},\alpha_j,\alpha_k) - \cL(u,u_k,u_i,u_{ki},\alpha_k,\alpha_i) ,
\end{split}
\end{align*}
is critical with respect to variations of the fields $u, u_i, \dots, u_{ijk}$~\cite{boll2014integrability}.}

Criticality with respect to variations of the fields leads to a set of generalised Euler--Lagrange equations, which we call \emph{corner equations}. (They are also known as \emph{multiform Euler--Lagrange equations} or \emph{multi-time Euler--Lagrange equations}.) For the Lagrangians usually given for the ABS equations, the corner equations are linear combinations of the three-leg forms without $2\pi {\rm i}$ terms, as in equation~\eqref{three-leg-additive-0}. To recover the three-leg equation modulo $2 \pi {\rm i}$, as in equation~\eqref{three-leg-additive-2pii}, we introduced an extended action in Part I:
\begin{align*}
 S^{\vec \Theta,\vec \Xi} &=
 S + 2 \pi {\rm i} ( \Theta u + \Theta_i u_i + \Theta_j u_j + \Theta_k u_k +\Theta_{ij} u_{ij} + \Theta_{jk} u_{jk} + \Theta_{ki} u_{ki} + \Theta_{ijk} u_{ijk}) \\
 &\quad{} + 2 \pi {\rm i} (\Xi_i \alpha_i + \Xi_j \alpha_j + \Xi_k \alpha_k),
\end{align*}
where $\Theta, \Theta_i, \dots, \Theta_{ijk}$ are integers associated to the vertices of the cube, and $\Xi_i$, $\Xi_j$, $\Xi_k$ are integers associated to the lattice directions.

The terms in $S^{\vec \Theta,\vec \Xi}$ involving $\Xi_i$, $\Xi_j$, $\Xi_k$ do not affect the corner equations, but they are key to the second variational principle in Lagrangian multiform theory: the action is invariant with respect to changes in the surface. This is equivalent to the property that the action on each elementary cube vanishes, i.e., \smash{$S^{\vec \Theta,\vec \Xi} = 0$} on solutions. \big(Except in some cases we only have \smash{$S^{\vec \Theta,\vec \Xi} \equiv 0 \mod 4 \pi^2$}.\big) Since this action can be thought of as the discrete exterior derivative of the Lagrangian 2-form, it is also referred to as the \emph{closure relation}.

Similarly, for the equations H1, A1$_{\delta=0}$ and Q1$_{\delta=0}$ with a three-leg form of type~\eqref{three-leg-additive-0}, we consider
\begin{align*}
 S^{\vec \Xi} &= S + 2 \pi {\rm i} (\Xi_i \alpha_i + \Xi_j \alpha_j + \Xi_k \alpha_k)
\end{align*}
in order to have a closure relation.

\subsection{Lagrangian leg functions}

The discrete Lagrangian 2-forms we consider in this work are linear combinations of functions of two lattice sites only.
After a suitable transformation of the variable $u$, there exist functions~$L$ and~$\Lambda$ such that the leg functions $\psi$ and $\phi$ can be expressed as
\begin{align*}
	&\psi(u,u_i, \alpha)= \pdv{}{u} L(u,u_i,\alpha_i) , \qquad
	\phi(u,u_{ij}, \alpha_i - \alpha_j)= \pdv{}{u} \Lambda(u,u_{ij},\alpha_i-\alpha_j) ,
\end{align*}
and also as
\begin{align*}
	&\pdv{}{u_i} L(u,u_i,\alpha_i)= \psi(u_i,u,\alpha_i) , \qquad
	\pdv{}{u_{ij}} \Lambda(u,u_{ij},\alpha_i-\alpha_j)= \phi(u_{ij},u, \alpha_i-\alpha_j) .
\end{align*}

A key element in the ABS classification is provided by the biquadratics associated to a quad equation, i.e., the polynomials $h$ and $g$ satisfying
 \begin{align*}
 & h(u,u_i,\alpha_i)= \frac{1}{k(\alpha_i,\alpha_j)} \biggl( Q \pdv{^2 Q}{u_j \partial u_{ij}} - \pdv{Q}{u_j}\pdv{Q}{u_{ij}} \biggr) , \\
 & g(u,u_{ij},\alpha_i-\alpha_j)= \frac{1}{k(\alpha_i,\alpha_j)} \biggl( Q \pdv{^2 Q}{u_i \partial u_j} - \pdv{Q}{u_i}\pdv{Q}{u_j} \biggr) ,
 \end{align*}
 where $k$ is a skew-symmetric function, chosen such that $h$ only depends on the indicated lattice parameter. They are related to the $L$ and $\Lambda$ by
\cite[Lemma 3]{bobenko2010lagrangian}{\samepage
 \begin{align}
	 &\pdv{L(u,u_i,\alpha_i)}{\alpha_i}\equiv \log h(u,u_i,\alpha_i) + \kappa(u) + \kappa(u_i) + c(\alpha_i) \mod 2 \pi {\rm i} \label{L-log(h)}\\
	 &\pdv{\Lambda(u,u_{ij},\alpha_i-\alpha_j)}{\alpha_i}\equiv \log g(u,u_{ij},\alpha_i-\alpha_j) + \kappa(u) + \kappa(u_{ij}) - \gamma(\alpha_i-\alpha_j) \mod 2 \pi {\rm i} , \label{Lambda-log(g)}
 \end{align}
 for some functions $\kappa$, $c$, $\gamma$.}

The quad equation $Q_{ij} = 0$ implies the following identities involving the biquadratics (see~\cite[Lemma 1]{bobenko2010lagrangian} and~\cite[Proposition 15]{adler2003classification}):
\begin{align}
	h(u,u_i,\alpha_i)h(u_{ij},u_j,\alpha_i)
 &= h(u,u_j,\alpha_j)h(u_{ij},u_i,\alpha_j) \nonumber \\
 &= g(u,u_{ij},\alpha_i-\alpha_j)g(u_i,u_j,\alpha_i-\alpha_j) .\label{biquadratric-quad}
\end{align}
Similarly, the tetrahedron equations
\[ T(u,u_{ij},u_{jk},u_{ki},\alpha_i,\alpha_j,\alpha_k)=0 \qquad \text{and} \qquad T(u_{ijk},u_{k},u_{i},u_{j},\alpha_i,\alpha_j,\alpha_k)=0\]
imply
\begin{align}
	g(u,u_{ij},\alpha_i-\alpha_j)g(u_{ki},u_{jk},\alpha_i-\alpha_j) & = g(u,u_{jk},\alpha_j-\alpha_k)g(u_{ki},u_{ij},\alpha_k-\alpha_i) \nonumber \\
	& = g(u,u_{ki},\alpha_k-\alpha_i)g(u_{ij},u_{jk},\alpha_k-\alpha_i)\label{biquadratric-tetra}
\end{align}
and
\begin{align}
	g(u_{ijk},u_{k},\alpha_i-\alpha_j)g(u_{j},u_{i},\alpha_i-\alpha_j) & = g(u_{ijk},u_{i},\alpha_j-\alpha_k)g(u_{j},u_{k},\alpha_k-\alpha_i) \nonumber \\
	& = g(u_{ijk},u_{j},\alpha_k-\alpha_i)g(u_{k},u_{i},\alpha_k-\alpha_i),\label{biquadratric-tetra-inverted}
\end{align}
respectively.

\subsection{Cross 2-form: tetrahedron equations are variational}
\label{sec-cross}

We will show that the tetrahedron equations arise as corner equations of the following 2-form, which we call the \emph{cross Lagrangian},
\[
 \cL_\cross(u,u_i,u_j,u_{ij},\alpha_i-\alpha_j) = \Lambda(u_i, u_j,\alpha_i-\alpha_j) - \Lambda(u, u_{ij},\alpha_i-\alpha_j) .
\]
The name for this Lagrangian is inspired on its leg structure, which is shown in Figure~\ref{fig-cross}\,(a). Note that the lattice is a bipartite graph, we can colour the vertices such that each edge links a black vertex with white vertex. Then each term of the cross Lagrangian involves two lattice sites of the same colour, so the cross Lagrangian (and its action) can be decomposed into one contribution from the black graph, and one from the white graph. These could be studied separately, as was done in~\cite{bobenko2012discrete, bobenko2010lagrangian} in the more general context of Laplace-type equations on a~bipartite quad graph.

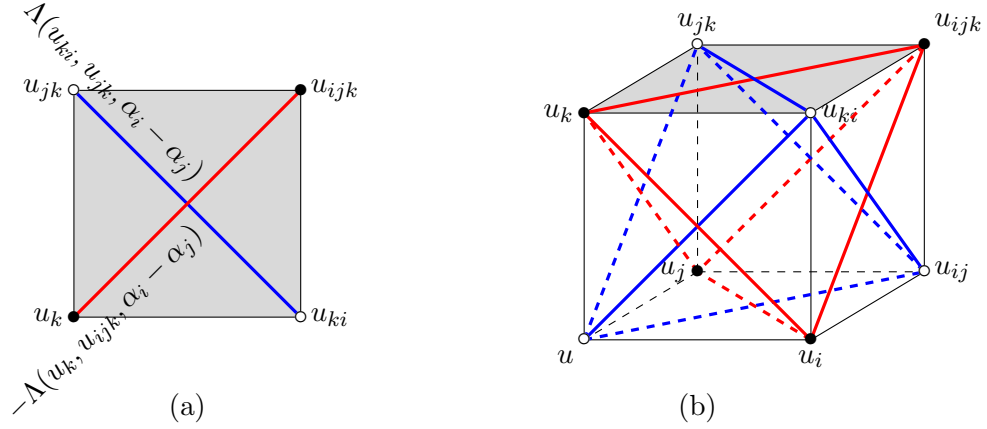
\begin{figure}[t]
\centering
\begin{tikzpicture}[scale=3]
	\begin{scope}[shift=({1.75,0})]
 \fill[gray!30] (0,0) -- (1,0) -- (1,1) -- (0,1) -- cycle;
		\draw (0,0) node[left] {$u_k$} -- (1,0) node[right] {$u_{ki}$} -- (1,1) node[right] {$u_{ijk}$} -- (0,1) node[left] {$u_{jk}$} -- cycle;
	 \draw[very thick, blue] (0,1) -- node[above left, rotate=-45, black] {\ $\Lambda(u_{ki},u_{jk}, \alpha_i-\alpha_j)$} (1,0);
		\draw[very thick, red] (1,1) -- node[below left, rotate=45, black] {\ $-\Lambda(u_k,u_{ijk}, \alpha_i-\alpha_j)$} (0,0);
 \node at (.5,-.4) {(a)};
 \node at (0,0) {$\bullet$};
 \node[white] at (1,0) {$\bullet$};
 \node[white] at (0,1) {$\bullet$};
 \node at (1,0) {$\circ$};
 \node at (0,1) {$\circ$};
 \node at (1,1) {$\bullet$};
	\end{scope}
	\begin{scope}[shift=({4,-.1}),scale=1, y={(5mm, 3mm)}, z={(0cm,1cm)}]
 \fill[gray!30] (0,0,1) -- (1,0,1) -- (1,1,1) -- (0,1,1) -- cycle;
		\draw[every edge/.append style={dashed}]
		(1,0,0) node[below] {$u_i$} -- (0,0,0) node[below left] {$u$} -- (0,0,1) node[left] {$u_k$} -- (0,1,1) node[above] {$u_{jk}$} -- (1,1,1) node[above right] {$u_{ijk}$} -- (1,1,0) node[right] {$u_{ij}$}-- cycle -- (1,0,1) node[right] {$u_{ki}$}-- (0,0,1)
		(1,0,1) -- (1,1,1)
		(0,1,0) edge (0,0,0) edge (0,1,1) edge (1,1,0) node[left] {$u_j$};
		
		\draw[very thick, blue, dashed] (0,0,0) -- (0,1,1) ;
		\draw[very thick, blue] (0,0,0) -- (1,0,1) ;
		\draw[very thick, blue, dashed] (0,0,0) -- (1,1,0) ;
		
		\draw[very thick, red] (1,1,1) -- (0,0,1) ;
		\draw[very thick, red] (1,1,1) -- (1,0,0) ;
		\draw[very thick, red, dashed] (1,1,1) -- (0,1,0) ;

 	\draw[very thick, blue, dashed] (1,1,0) -- (0,1,1) ;
 	\draw[very thick, blue] (0,1,1) -- (1,0,1) ;
 	\draw[very thick, blue] (1,0,1) -- (1,1,0) ;
 	
 	\draw[very thick, red, dashed] (1,0,0) -- (0,1,0) ;
 	\draw[very thick, red, dashed] (0,1,0) -- (0,0,1) ;
 	\draw[very thick, red] (0,0,1) -- (1,0,0) ;
 \node at (1,-1) {(b)};
 \node[white] at (0,0,0) {$\bullet$};
 \node at (0,0,0) {$\circ$};
 \node at (1,0,0) {$\bullet$};
 \node at (0,1,0) {$\bullet$};
 \node at (0,0,1) {$\bullet$};
 \node[white] at (1,1,0) {$\bullet$};
 \node[white] at (0,1,1) {$\bullet$};
 \node[white] at (1,0,1) {$\bullet$};
 \node at (1,1,0) {$\circ$};
 \node at (0,1,1) {$\circ$};
 \node at (1,0,1) {$\circ$};
 \node at (1,1,1) {$\bullet$};
	\end{scope}	
\end{tikzpicture}
\caption[.]{(a) The leg structure of a single Lagrangian $\cL_\cross(u_k,u_{ki},u_{jk},u_{ijk},\alpha_i,\alpha_j)$. (b) The leg structure for the action on an elementary cube of the cross 2-form $\cL_\cross$.}\label{fig-cross}
% added "[.]" to fix error, inspired by https://stackoverflow.com/questions/2716227/adding-an-equation-or-formula-to-a-figure-caption-in-latex
\end{figure}

Consider the action of $\cL_\cross$ over an elementary cube
\begin{align}
 S_\cross &= \Lambda(u_{ki}, u_{jk},\alpha_i-\alpha_j) - \Lambda(u_k, u_{ijk},\alpha_i-\alpha_j) \notag \\
 &\quad{} - \Lambda(u_i, u_j,\alpha_i-\alpha_j) + \Lambda(u, u_{ij},\alpha_i-\alpha_j) + \cyclic \notag\\
 & = \Lambda(u, u_{ij},\alpha_i-\alpha_j) - \Lambda(u_i, u_j,\alpha_i-\alpha_j) - \reverse + \cyclic ,\label{Scross}
\end{align}
where $\cyclic$ denotes the expressions obtained from the preceding one by cyclic permutations, and~\smash{\raisebox{-1pt}{$\reverse$}} denotes the point inversion of the preceding terms, i.e., the expression obtained by~interchanging $u \leftrightarrow u_{ijk}$, $u_i \leftrightarrow u_{jk}$, etc.
Consider also he extended actions
\begin{align*}
 &S^{\vec \Xi}_\cross= S_\cross + 2 \pi {\rm i} (\Xi_i \alpha_i + \Xi_j \alpha_j + \Xi_k \alpha_k) , \\
 &S^{\vec \Theta,\vec \Xi}_\cross= S_\cross + 2 \pi {\rm i} ( \Theta u + \Theta_i u_i + \Theta_j u_j + \Theta_k u_k +\Theta_{ij} u_{ij} + \Theta_{jk} u_{jk} + \Theta_{ki} u_{ki} + \Theta_{ijk} u_{ijk}) \\
 &\hphantom{S^{\vec \Theta,\vec \Xi}_\cross=}{} + 2 \pi {\rm i} (\Xi_i \alpha_i + \Xi_j \alpha_j + \Xi_k \alpha_k).
\end{align*}

\begin{Proposition}
 \label{prop-cross}
 For ${\rm H}1$, ${\rm A}1_{\delta=0}$ and ${\rm Q}1_{\delta=0}$, the tetrahedron equations
 \[ T(u,u_{ij},u_{jk},u_{ki},\alpha_i,\alpha_j,\alpha_k) = 0 ,
 \qquad
 T(u_{ijk},u_{k},u_{i},u_{j},\alpha_i,\alpha_j,\alpha_k) = 0 \]
 are satisfied if and only if there exist $\Xi_i,\Xi_j,\Xi_k \in \mathbb{Z}$ such that
 \begin{alignat*}{6}
 &\pdv{S^{\vec \Xi}_\cross}{u} = 0 , \qquad&&
 \pdv{S^{\vec \Xi}_\cross}{u_i} = 0 , \qquad&&
 \pdv{S^{\vec \Xi}_\cross}{u_j} = 0 , \qquad&&
 \pdv{S^{\vec \Xi}_\cross}{u_k} = 0 ,&& &\\
 & && \pdv{S^{\vec \Xi}_\cross}{u_{jk}} = 0 , \qquad&&
 \pdv{S^{\vec \Xi}_\cross}{u_{ki}} = 0 , \qquad&&
 \pdv{S^{\vec \Xi}_\cross}{u_{ij}} = 0 , \qquad&&
 \pdv{S^{\vec \Xi}_\cross}{u_{ijk}} = 0 ,&
 \end{alignat*}
 and
 \begin{equation*}
 \pdv{S^{\vec \Xi}_\cross}{\alpha_i} = 0 , \qquad
 \pdv{S^{\vec \Xi}_\cross}{\alpha_j} = 0 , \qquad
 \pdv{S^{\vec \Xi}_\cross}{\alpha_k} = 0 .
 \end{equation*}

 For ${\rm H}2$, ${\rm H}3$, ${\rm A}1_{\delta\neq0}$, ${\rm A}2$, ${\rm Q}1_{\delta\neq0}$, ${\rm Q}2$, ${\rm Q}3$, the tetrahedron equations are satisfied if and only if there exist $\Theta,\Theta_i,\dots,\Theta_{ijk} \in \mathbb{Z}$ and $\Xi_i,\Xi_j,\Xi_k \in \mathbb{Z}$ such that
 \begin{alignat*}{6}
 &\pdv{S^{\vec \Theta,\vec \Xi}_\cross}{u} = 0 ,\qquad &&
 \pdv{S^{\vec \Theta,\vec \Xi}_\cross}{u_i} = 0 ,\qquad &&
 \pdv{S^{\vec \Theta,\vec \Xi}_\cross}{u_j} = 0 ,\qquad &&
 \pdv{S^{\vec \Theta,\vec \Xi}_\cross}{u_k} = 0 ,&& & \\
 & && \pdv{S^{\vec \Theta,\vec \Xi}_\cross}{u_{jk}} = 0 ,\qquad &&
 \pdv{S^{\vec \Theta,\vec \Xi}_\cross}{u_{ki}} = 0 ,\qquad &&
 \pdv{S^{\vec \Theta,\vec \Xi}_\cross}{u_{ij}} = 0 ,\qquad &&
 \pdv{S^{\vec \Theta,\vec \Xi}_\cross}{u_{ijk}} = 0 , &
 \end{alignat*}
 and
 \begin{equation*}
 \pdv{S^{\vec \Theta,\vec \Xi}_\cross}{\alpha_i} = 0 , \qquad
 \pdv{S^{\vec \Theta,\vec \Xi}_\cross}{\alpha_j} = 0 , \qquad
 \pdv{S^{\vec \Theta,\vec \Xi}_\cross}{\alpha_k} = 0 .
 \end{equation*}
\end{Proposition}
\begin{proof}\allowdisplaybreaks
 We give the proof for the case of H2, H3, A1$_{\delta\neq0}$, A2, Q1$_{\delta\neq0}$, Q2, Q3. The argument for H1, A1$_{\delta=0}$ and Q1$_{\delta=0}$ is obtained from this by setting $\vec \Theta = \vec 0$.
 In the partial derivatives of~\eqref{Scross}, we recognise three-leg forms of the tetrahedron equations,~\eqref{tetra-u}--\eqref{tetra-ui}. We have
 \begin{gather}
 \pdv{S^{\vec \Theta,\vec \Xi}_\cross}{u}
 = \phi(u,u_{ij},\alpha_i - \alpha_j) + \cyclic + 2 \Theta \pi {\rm i}
 = \mathcal{T}^{(u)} + 2 \Theta \pi {\rm i} , \label{cross-corners-u} \\
 \pdv{S^{\vec \Theta,\vec \Xi}_\cross}{u_i}
 = -\phi(u_i,u_j,\alpha_i - \alpha_j) - \phi(u_i,u_k,\alpha_k - \alpha_i) - \phi(u_i,u_{ijk},\alpha_j - \alpha_k) + 2 \Theta_i \pi {\rm i} \notag\\
 \hphantom{\pdv{S^{\vec \Theta,\vec \Xi}_\cross}{u_i}}{} = -\mathcal{T}^{(u_i)} + 2 \Theta_i \pi {\rm i} , \\
 \pdv{S^{\vec \Theta,\vec \Xi}_\cross}{u_{ij}}
 = \phi(u_{ij},u,\alpha_i-\alpha_j) + \phi(u_{ij},u_{jk},\alpha_k - \alpha_i) + \phi(u_{ij},u_{ki},\alpha_j-\alpha_k) + 2 \Theta_{ij} \pi {\rm i} \notag\\
 \hphantom{\pdv{S^{\vec \Theta,\vec \Xi}_\cross}{u_{ij}}}{} = \mathcal{T}^{(u_{ij})} + 2 \Theta_{ij} \pi {\rm i} , \\
 \pdv{S^{\vec \Theta,\vec \Xi}_\cross}{u_{ijk}}
 = -\phi(u_{ijk},u_k,\alpha_i-\alpha_j) - \phi(u_{ijk},u_j,\alpha_k - \alpha_i) - \phi(u_{ijk},u_{i},\alpha_j-\alpha_k) + 2 \Theta_{ijk} \pi {\rm i} \notag\\
 \hphantom{\pdv{S^{\vec \Theta,\vec \Xi}_\cross}{u_{ijk}}}{} = -\mathcal{T}^{(u_{ijk})} + 2 \Theta_{ijk} \pi {\rm i} , \label{cross-corners-uijk}
 \end{gather}
 and expressions obtained from these by cyclic permutation of the indices.
 According to Proposition~\ref{prop-three-leg-quad}, there exist $\Theta, \Theta_i, \dots, \Theta_{ijk} \in \Z$ such that these expressions are zero if and only if the multi-affine tetrahedron equation $T = 0$ holds.

 To prove the second claim, we use equations~\eqref{L-log(h)} and~\eqref{Lambda-log(g)} to compute
 \begin{align*}
 \pdv{S^{\vec \Theta,\vec \Xi}_\cross}{\alpha_i}
 &= \log(g(u,u_{ij},\alpha_i-\alpha_j)) - \log(g(u_i,u_j,\alpha_i-\alpha_j)) - \reverse + \cyclic + 2 \Xi_i \pi {\rm i} \\
 &\equiv \log \biggl( \frac{g(u,u_{ij},\alpha_i-\alpha_j) g(u_{ki},u_{jk},\alpha_i-\alpha_j)}{g(u_k,u_{ijk},\alpha_i-\alpha_j) g(u_i,u_j,\alpha_i-\alpha_j)} \biggr) + \cyclic \mod 2 \pi {\rm i} .
 \end{align*}
 By virtue of equations~\eqref{biquadratric-tetra} and~\eqref{biquadratric-tetra-inverted}, the logarithm in this expression vanishes on the tetrahedron equations, so given a solution to the tetrahedron equations, we can choose $\Xi_i \in \Z$ such that\looseness=-1
 \[
 \pdv{S^{\vec \Theta,\vec \Xi}_\cross}{\alpha_i} = 0.\tag*{\qed}
 \]\renewcommand{\qed}{}
\end{proof}

The closure of the cross-multiform is closely related to the star-triangle relation: the closure relation splits into contributions on the black and white parts of the bipartite graph, each of which encodes a star-triangle relation, see~\cite[Theorem~2]{bobenko2010lagrangian}.
More about the relation between ABS equations and star-triangle relations can be found, for example, in~\cite{bazhanov2016quasiclassical}. Note that our strategy of including integer fields to account for branch cuts has previously been used in the context of star-triangle relations~\cite{kels2021interaction, kels2023Twocomponent}.

\subsection{Trident 2-form: quad equations are variational}
%\label{sec-trident}

One of the main results of Part~I was that the quad equations of the ABS list are variational. We showed that the trident Lagrangian
\[
 \cL_\trileg(u,u_i,u_j,u_{ij},\alpha_i,\alpha_j) := L(u,u_i,\alpha_i) - L(u,u_j,\alpha_j) - \Lambda(u,u_{ij}, \alpha_i-\alpha_j) ,
\]
has corner equations that produce the quad equations directly (in their three-leg form).
This is illustrated in Figure~\ref{fig-trid}, where it can be seen that at each vertex there are three legs contributing to the action around the cube, which either lie in a single quad, or span a tetrahedron.

\begin{figure}[t]
\centering
\begin{tikzpicture}[scale=3]
	\begin{scope}[shift=({1.75,0})]
 \fill[gray!30] (0,0) -- (1,0) -- (1,1) -- (0,1) -- cycle;
		\draw (0,0) node[below left] {$u_k$} -- (1,0) node[below right] {$u_{ki}$} -- (1,1) node[above right] {$u_{ijk}$} -- (0,1) node[above left] {$u_{jk}$} -- cycle;
		\draw[very thick, blue] (0,0) -- node[below, black] {$L(u_k,u_{ki},\alpha_i)$} (1,0);
		\draw[very thick, red] (0,1) -- node[above, rotate=90, black] {$-L(u_k,u_{jk},\alpha_j)$} (0,0);
		\draw[very thick, red] (1,1) -- node[above, rotate=45, black] {\ $-\Lambda(u_k,u_{ijk}, \alpha_i-\alpha_j)$} (0,0);
 \node at (.5,-.4) {(a)};
 \node at (0,0) {$\bullet$};
 \node at (1,0) {$\bullet$};
 \node at (0,1) {$\bullet$};
 \node at (1,1) {$\bullet$};
	\end{scope}
	\begin{scope}[shift=({4,-.1}),scale=1, y={(5mm, 3mm)}, z={(0cm,1cm)}]
 \fill[gray!30] (0,0,1) -- (1,0,1) -- (1,1,1) -- (0,1,1) -- cycle;
		\draw[every edge/.append style={dashed}]
		(1,0,0) node[below] {$u_i$} -- (0,0,0) node[below left] {$u$} -- (0,0,1) node[left] {$u_k$} -- (0,1,1) node[above] {$u_{jk}$} -- (1,1,1) node[above right] {$u_{ijk}$} -- (1,1,0) node[right] {$u_{ij}$}-- cycle -- (1,0,1) node[right] {$u_{ki}$}-- (0,0,1)
		(1,0,1) -- (1,1,1)
		(0,1,0) edge (0,0,0) edge (0,1,1) edge (1,1,0) node[left] {$u_j$};
		
		\draw[very thick, red, dashed] (1,1,0) -- (0,1,0);
		\draw[very thick, blue, dashed] (0,1,0) -- (0,1,1);
		
		\draw[very thick, blue, dashed] (0,0,0) -- (0,1,1) ;
		\draw[very thick, blue] (0,0,0) -- (1,0,1) ;
		\draw[very thick, blue, dashed] (0,0,0) -- (1,1,0) ;
		
		\draw[very thick, red] (1,1,1) -- (0,0,1) ;
		\draw[very thick, red] (1,1,1) -- (1,0,0) ;
		\draw[very thick, red, dashed] (1,1,1) -- (0,1,0) ;
		
		\draw[very thick, blue] (1,0,0) -- (1,1,0);
		\draw[very thick, red] (0,1,1) -- (0,0,1);
		\draw[very thick, blue] (0,0,1) -- (1,0,1);
		\draw[very thick, red] (1,0,1) -- (1,0,0);
 \node at (1,-1) {(b)};
 \node at (0,0,0) {$\bullet$};
 \node at (1,0,0) {$\bullet$};
 \node at (0,1,0) {$\bullet$};
 \node at (0,0,1) {$\bullet$};
 \node at (1,1,0) {$\bullet$};
 \node at (0,1,1) {$\bullet$};
 \node at (1,0,1) {$\bullet$};
 \node at (1,1,1) {$\bullet$};
	\end{scope}	
\end{tikzpicture}
\caption[.]{(a) The leg structure of a single Lagrangian $\cL_\trileg(u_k,u_{ki},u_{jk},u_{ijk},\alpha_i,\alpha_j)$. (b) The leg structure for the action on an elementary cube of the trident 2-form $\cL_\trileg$.
}
\label{fig-trid}
% added "[.]" to fix error, inspired by https://stackoverflow.com/questions/2716227/adding-an-equation-or-formula-to-a-figure-caption-in-latex
\end{figure}

\begin{Proposition} \label{prop-trileg}\samepage
 For ${\rm H}2$, ${\rm H}3$, ${\rm A}1_{\delta\neq0}$, ${\rm A}2$, ${\rm Q}1_{\delta\neq0}$, ${\rm Q}2$, ${\rm Q}3$, the quad equations are satisfied if and only if there exist $\Theta,\Theta_i,\dots,\Theta_{ijk} \in \mathbb{Z}$ and $\Xi_i,\Xi_j,\Xi_k \in \mathbb{Z}$ such that
 \begin{alignat}{6}
 &\pdv{S^{\vec \Theta,\vec \Xi}_\trileg}{u} = 0 , \qquad &&
 \pdv{S^{\vec \Theta,\vec \Xi}_\trileg}{u_i} = 0 , \qquad &&
 \pdv{S^{\vec \Theta,\vec \Xi}_\trileg}{u_j} = 0 , \qquad &&
 \pdv{S^{\vec \Theta,\vec \Xi}_\trileg}{u_k} = 0 , && &\nonumber\\
 & && \pdv{S^{\vec \Theta,\vec \Xi}_\trileg}{u_{jk}} = 0 , \qquad &&
 \pdv{S^{\vec \Theta,\vec \Xi}_\trileg}{u_{ki}} = 0 , \qquad &&
 \pdv{S^{\vec \Theta,\vec \Xi}_\trileg}{u_{ij}} = 0 , \qquad &&
 \pdv{S^{\vec \Theta,\vec \Xi}_\trileg}{u_{ijk}} = 0 , & \label{trid-corners=0}
 \end{alignat}
 and
 \begin{equation}
 \label{trid-alphas}
 \pdv{S^{\vec \Theta,\vec \Xi}_\trileg}{\alpha_i} = 0 , \qquad
 \pdv{S^{\vec \Theta,\vec \Xi}_\trileg}{\alpha_j} = 0 , \qquad
 \pdv{S^{\vec \Theta,\vec \Xi}_\trileg}{\alpha_k} = 0 .
 \end{equation}

 For ${\rm H}1$, ${\rm A}1_{\delta=0}$ and ${\rm Q}1_{\delta=0}$, the quad equations are satisfied if and only if there exist integers $\Xi_i,\Xi_j,\Xi_k \in \mathbb{Z}$ such that equations~\eqref{trid-corners=0} and~\eqref{trid-alphas} hold for $S^{\vec \Xi}$.
\end{Proposition}

\begin{proof}
 We identify the corner equations with the three-leg forms of the quad equations~\eqref{quad-ui}--\eqref{quad-uj} and tetrahedron equations~\eqref{tetra-u} and~\eqref{tetra-uijk} by observing that
 \begin{equation}
 \label{trid-corners}
 \pdv{S_\trileg^{\vec \Theta, \vec \Xi}}{u} = \mathcal{T}^{(u)} + 2 \Theta \pi {\rm i} , \qquad
 \pdv{S_\trileg^{\vec \Theta, \vec \Xi}}{u_i} = \mathcal{Q}^{(u_i)}_{jk} + 2 \Theta_{i} \pi {\rm i} ,
 \qquad \text{etc.,}
 \end{equation}
 where \smash{$\mathcal{Q}^{(u_i)}_{jk}\! = \psi(u_i,u_{ij},\alpha_j) - \psi(u_i,u_{ki},\alpha_k) - \phi(u_i,u_{ijk},\alpha_j-\alpha_k)$} is the three leg form based~at~$u_i$ in the quad oriented along the $j,k$-directions. Since the quad equations imply the tetrahedron equations, equations~\eqref{trid-corners} imply that the quad equations are equivalent to the system~\eqref{trid-corners=0}.

 Similar to Proposition~\ref{prop-cross}, equations~\eqref{trid-alphas} can be proved using the biquadratics associated to each quad equation. Details can be found in~\cite{richardson2025discrete1}.
\end{proof}

\subsection{Cross-square 2-form}
%\label{sec-cross-square}

We introduce a second 2-form that produces a system of corner equations equivalent to the quad equations. The function defining this 2-form was introduced in~\cite{bobenko2010lagrangian}, where it was studied on a~single quad. To our knowledge, it was never before considered as a Lagrangian multiform. It is given by
\begin{align*}
\begin{split}
		\cL_\crosssquare:={}& L(u,u_i,\alpha_i) + L(u_{ij},u_j, \alpha_i) - L(u,u_j,\alpha_j) - L(u_{ij},u_i, \alpha_j) \\
		&{}{-} \Lambda(u,u_{ij},\alpha_i-\alpha_j) - \Lambda(u_i,u_j,\alpha_i-\alpha_j).
\end{split}
\end{align*}
Inspired by its leg structure, illustrated in Figure~\ref{fig-crosssquare}, we call it the \emph{cross-square Lagrangian}.

\begin{figure}[t]
\centering
\begin{tikzpicture}[scale=3]
	\begin{scope}[shift=({1.75,0})]
 \fill[gray!30] (0,0) -- (1,0) -- (1,1) -- (0,1) -- cycle;
		\draw[thick, dashed] (0,0) node[left] {$u_k$} -- (1,0) node[right] {$u_{ki}$} -- (1,1) node[right] {$u_{ijk}$} -- (0,1) node[left] {$u_{jk}$} -- cycle;
		\draw[very thick, blue] (0,0) -- (1,0);
 \draw[very thick, blue] (0,1) -- (1,1);
		\draw[very thick, red] (0,1) -- (0,0);
 \draw[very thick, red] (1,1) -- (1,0);
		\draw[very thick, red] (1,1) -- (0,0);
 \draw[very thick, red] (1,0) -- (0,1);
 \node at (.5,-.4) {(a)};
 \node at (0,0) {$\bullet$};
 \node at (1,0) {$\bullet$};
 \node at (0,1) {$\bullet$};
 \node at (1,1) {$\bullet$};
	\end{scope}
	\begin{scope}[shift=({4,-.1}),scale=1, y={(5mm, 3mm)}, z={(0cm,1cm)}]
 \fill[gray!30] (0,0,1) -- (1,0,1) -- (1,1,1) -- (0,1,1) -- cycle;
		\draw[every edge/.append style={dashed}]
		(1,0,0) node[below] {$u_i$} -- (0,0,0) node[below left] {$u$} -- (0,0,1) node[left] {$u_k$} -- (0,1,1) node[above] {$u_{jk}$} -- (1,1,1) node[above right] {$u_{ijk}$} -- (1,1,0) node[right] {$u_{ij}$}-- cycle -- (1,0,1) node[right] {$u_{ki}$}-- (0,0,1)
		(1,0,1) -- (1,1,1)
		(0,1,0) edge (0,0,0) edge (0,1,1) edge (1,1,0) node[left] {$u_j$};
		
		\draw[very thick, red, dashed, double] (1,1,0) -- (0,1,0);
		\draw[very thick, blue, dashed, double] (0,1,0) -- (0,1,1);
		
		\draw[very thick, blue, dashed] (0,0,0) -- (0,1,1) ;
		\draw[very thick, blue, dashed] (0,1,0) -- (0,0,1) ;
		\draw[very thick, blue] (0,0,0) -- (1,0,1) ;
		\draw[very thick, blue] (1,0,0) -- (0,0,1) ;
		\draw[very thick, blue, dashed] (0,0,0) -- (1,1,0) ;
		\draw[very thick, blue, dashed] (1,0,0) -- (0,1,0) ;
		
		\draw[very thick, red] (1,1,1) -- (0,0,1) ;
		\draw[very thick, red] (1,0,1) -- (0,1,1) ;
		\draw[very thick, red] (1,1,1) -- (1,0,0) ;
		\draw[very thick, red] (1,1,0) -- (1,0,1) ;
		\draw[very thick, red, dashed] (1,1,1) -- (0,1,0) ;
		\draw[very thick, red, dashed] (1,1,0) -- (0,1,1) ;
		
		\draw[very thick, blue, double] (1,0,0) -- (1,1,0);
		\draw[very thick, red, double] (0,1,1) -- (0,0,1);
		\draw[very thick, blue, double] (0,0,1) -- (1,0,1);
		\draw[very thick, red, double] (1,0,1) -- (1,0,0);
 \node at (1,-1) {(b)};
 \node at (0,0,0) {$\bullet$};
 \node at (1,0,0) {$\bullet$};
 \node at (0,1,0) {$\bullet$};
 \node at (0,0,1) {$\bullet$};
 \node at (1,1,0) {$\bullet$};
 \node at (0,1,1) {$\bullet$};
 \node at (1,0,1) {$\bullet$};
 \node at (1,1,1) {$\bullet$};
	\end{scope}	
\end{tikzpicture}
\caption[.]{(a) The leg structure of a single Lagrangian $\cL_\crosssquare(u_k,u_{ki},u_{jk},u_{ijk},\alpha_i,\alpha_j)$. (b) The leg structure for the action on an elementary cube of the cross-square 2-form $\cL_\crosssquare$.}
\label{fig-crosssquare}
% added "[.]" to fix error, inspired by https://stackoverflow.com/questions/2716227/adding-an-equation-or-formula-to-a-figure-caption-in-latex
\end{figure}

The action over an elementary cube of the cross-square Lagrangian is
\begin{align}
 S_\crosssquare &= 2 L(u_i,u_{ij},\alpha_j) + \Lambda(u,u_{ij},\alpha_i-\alpha_j) + \Lambda(u_i,u_j,\alpha_i-\alpha_j) - \reverse + \cyclic \notag \\
 &= 2 S_\trileg - S_\cross ,\label{Scrosssquare-tridcross}
\end{align}
so its corner equations consist of linear combinations of the corner equations of the trident and cross Lagrangians.

\begin{Proposition}
 %\label{prop-crosssquare}
 For ${\rm H}2$, ${\rm H}3$, ${\rm A}1_{\delta\neq0}$, ${\rm A}2$, ${\rm Q}1_{\delta\neq0}$, ${\rm Q}2$, ${\rm Q}3$, the quad equations are satisfied if and only if there exist $\Theta,\Theta_i,\dots,\Theta_{ijk} \in \mathbb{Z}$ and $\Xi_i,\Xi_j,\Xi_k \in \mathbb{Z}$ such that
 \begin{alignat}{6}
 &\pdv{S^{\vec \Theta,\vec \Xi}_\crosssquare}{u} = 0 , \qquad&&
 \pdv{S^{\vec \Theta,\vec \Xi}_\crosssquare}{u_i} = 0 , \qquad &&
 \pdv{S^{\vec \Theta,\vec \Xi}_\crosssquare}{u_j} = 0 , \qquad &&
 \pdv{S^{\vec \Theta,\vec \Xi}_\crosssquare}{u_k} = 0 , && &\nonumber\\
 & && \pdv{S^{\vec \Theta,\vec \Xi}_\crosssquare}{u_{jk}} = 0 , \qquad &&
 \pdv{S^{\vec \Theta,\vec \Xi}_\crosssquare}{u_{ki}} = 0 , \qquad &&
 \pdv{S^{\vec \Theta,\vec \Xi}_\crosssquare}{u_{ij}} = 0 , \qquad &&
 \pdv{S^{\vec \Theta,\vec \Xi}_\crosssquare}{u_{ijk}} = 0 ,& \label{crosssquare-corners=0}
 \end{alignat}
 and
 \begin{align}
 \label{crosssquare-alphas}
 \pdv{S^{\vec \Theta,\vec \Xi}_\crosssquare}{\alpha_i} = 0 ,\qquad
 \pdv{S^{\vec \Theta,\vec \Xi}_\crosssquare}{\alpha_j} = 0 , \qquad
 \pdv{S^{\vec \Theta,\vec \Xi}_\crosssquare}{\alpha_k} = 0 .
 \end{align}

 For ${\rm H}1$, ${\rm A}1_{\delta=0}$ and ${\rm Q}1_{\delta=0}$, the quad equations are satisfied if and only if there exist integers $\Xi_i,\Xi_j,\Xi_k \in \mathbb{Z}$ such that equations~\eqref{crosssquare-corners=0} and~\eqref{crosssquare-alphas} hold for $S^{\vec \Xi}$.
\end{Proposition}

\begin{proof}
 {\samepage From~\eqref{Scrosssquare-tridcross}, we deduce that the corner equations are the corresponding linear combination of the corner equations~\eqref{trid-corners} and~\eqref{cross-corners-u}--\eqref{cross-corners-uijk} of the trident and cross Lagrangian 2-forms. Hence,
\begin{alignat}{3}
 &\pdv{S_\crosssquare}{u} = \mathcal{T}^{(u)} + 2 \Theta \pi {\rm i} , \qquad && \pdv{S_\crosssquare}{u_{ijk}} = -\mathcal{T}^{(u_{ijk})} + 2 \Theta_{ijk} \pi {\rm i} , & \nonumber\\
 &\pdv{S_\crosssquare}{u_i} = 2\mathcal{Q}^{(u_i)}_{jk} + \mathcal{T}^{(u_i)} + 2 \Theta_i \pi {\rm i} ,\qquad && \pdv{S_\crosssquare}{u_{jk}} = -2\mathcal{Q}^{(u_{jk})}_{jk} - \mathcal{T}^{(u_{jk})} + 2 \Theta_{jk} \pi {\rm i} , & \label{crosssquare-corners}
 \end{alignat}
 where \smash{$\mathcal{Q}^{(u_i)}_{jk}\! = \psi(u_i,u_{ij},\alpha_j) - \psi(u_i,u_{ki},\alpha_k) - \phi(u_i,u_{ijk},\alpha_j-\alpha_k)$} is the three leg form based at~$u_i$ in the quad oriented along the $j,k$-directions and \smash{$\mathcal{Q}\raisebox{-1pt}{${}^{(u_jk)}_{jk}$}$} is as in equation~\eqref{quad-uij}.
 The remaining corner equations are obtained from~\eqref{crosssquare-corners} by cyclic permutation of the indices.}

 If multi-affine quad equations hold, then all of the three-leg forms involved in the system~\eqref{crosssquare-corners} are multiples of $2 \pi {\rm i}$, so the corner equations are satisfied. Conversely, if all equations of the system~\eqref{crosssquare-corners} hold, then in particular the tetrahedron equations in three-leg form, centred at~$u$ and $u_{ijk}$, are satisfied. Then the tetrahedron three-leg forms centred at $u_i$, $u_j$, $u_k$, $u_{ij}$, $u_{jk}$, $u_{ki}$ also vanish modulo $2 \pi {\rm i}$, so the remaining corner equations reduce to individual quad equations in three-leg form. These are equivalent to the multi-affine quad equations.

 Furthermore, using equations~\eqref{L-log(h)} and~\eqref{Lambda-log(g)}, we find
 \begin{align*}
 \pdv{\cL_\crosssquare}{\alpha_i}
 &= \log(h(u,u_i,\alpha_i)) + \log(h(u_{ij},u_j,\alpha_i)) \\
 &\quad - \log(g(u,u_{ij},\alpha_i-\alpha_j)) - \log(g(u_i,u_j,\alpha_i-\alpha_j)) + 2 c(\alpha_i) - 2 \gamma(\alpha_i-\alpha_j) \\
 &\equiv \log \biggl(\frac{h(u,u_i,\alpha_i) h(u_{ij},u_j,\alpha_i)}{g(u,u_{ij},\alpha_i-\alpha_j) g(u_i,u_j,\alpha_i-\alpha_j)} \biggr) + 2 c(\alpha_i) - 2 \gamma(\alpha_i-\alpha_j) \mod 2 \pi {\rm i} .
 \end{align*}
 From equation~\eqref{biquadratric-quad}, it now follows that, on solutions to the quad equations, \smash{$\pdv{S_\crosssquare}{\alpha_i}$} is a multiple of $2 \pi {\rm i} $, so there exists a $\Xi_i$ such that
 \[ \pdv{S_\crosssquare^{\vec \Theta, \vec \Xi}}{\alpha_i} = 0 . \tag*{\qed} \]
 \renewcommand{\qed}{}
\end{proof}

The following lemma states that, for the cross-square Lagrangian, we can also obtain the quad equation by taking partial derivatives of the Lagrangian on a single quad. Lagrangian multiform theory does not require this property, but it will be a useful tool in Section~\ref{sec-double0-ABS}.

\begin{Lemma}
 \label{lemma-crosssquare}
 For ${\rm H}2$, ${\rm H}3$, ${\rm A}1_{\delta\neq0}$, ${\rm A}2$, ${\rm Q}1_{\delta\neq0}$, ${\rm Q}2$, ${\rm Q}3$, the quad equation is satisfied if and only if there exist $\Theta,\Theta_i,\Theta_j,\Theta_{ij} \in \mathbb{Z}$ and $\Xi_i,\Xi_j \in \mathbb{Z}$ such that
 \begin{align*}
 \pdv{\cL_\crosssquare}{u} = 2 \Theta \pi {\rm i} , \qquad
 \pdv{\cL_\crosssquare}{u_i} = 2 \Theta_i \pi {\rm i} , \qquad
 \pdv{\cL_\crosssquare}{u_j} = 2 \Theta_j \pi {\rm i} , \qquad
 \pdv{\cL_\crosssquare}{u_{ij}} = 2 \Theta_{ij} \pi {\rm i} .
 \end{align*}

 For ${\rm H}1$, ${\rm A}1_{\delta=0}$ and ${\rm Q}1_{\delta=0}$, the quad equation is satisfied if and only if there exist $\Xi_i,\Xi_j \in \mathbb{Z}$ such that
 \begin{align*}
 \pdv{\cL_\crosssquare}{u} = 0 , \qquad
 \pdv{\cL_\crosssquare}{u_i} = 0 , \qquad
 \pdv{\cL_\crosssquare}{u_j} = 0 , \qquad
 \pdv{\cL_\crosssquare}{u_{ij}} = 0 .
 \end{align*}
\end{Lemma}
\begin{proof}
 We have
 \[ \pdv{\cL_\crosssquare}{u} = \psi(u,u_i,\alpha_i) - \psi(u,u_j,\alpha_j) - \phi(u, u_{ij}, \alpha_i - \alpha_j) \]
 and similar expressions for the derivatives with respect to $u_i$, $u_j$, and $u_{ij}$.
\end{proof}

\subsection{Triangle 2-form}
\label{sec-triangle}

Now we compare the above constructions to the known 2-form $\cL_\triang$ (which has been the standard choice in the theory of Lagrangian multiforms~\cite{lobb2009lagrangian}),
\[
 \cL_\triang(u,u_i,u_j,\alpha_i,\alpha_j) := L(u,u_i,\alpha_i) - L(u,u_j,\alpha_j) - \Lambda(u_i,u_j, \alpha_i-\alpha_j) .
\]
We review how this 2-form is critical on a set of equations weaker than the quad equations and we will relate it to our new 2-forms via $\cL_\triang = \cL_\trileg - \cL_\cross$.

\begin{figure}[t]
\centering
\begin{tikzpicture}[scale=3]
 \begin{scope}[shift=({1.75,0})]
 \fill[gray!30] (0,0) -- (1,0) -- (1,1) -- (0,1) -- cycle;
		\draw (0,0) node[below left] {$u_k$} -- (1,0) node[below right] {$u_{ki}$} -- (1,1) -- (0,1) node[above left] {$u_{jk}$} -- cycle;
 	\draw[very thick, blue] (0,0) -- node[below, black] {$L(u_k,u_{ki},\alpha_i)$} (1,0);
 	\draw[very thick, red] (0,1) -- node[above, rotate=90, black] {$-L(u_k,u_{jk},\alpha_j)$} (0,0);
 	\draw[very thick, red] (0,1) -- node[above, rotate=-45, black] {\ $-\Lambda(u_{ki},u_{jk}, \alpha_i-\alpha_j)$} (1,0);
 \node at (.5,-.4) {(a)};
 \node at (0,0) {$\bullet$};
 \node at (1,0) {$\bullet$};
 \node at (0,1) {$\bullet$};
	\end{scope}
	\begin{scope}[shift=({4,-.1}),scale=1, y={(5mm, 3mm)}, z={(0cm,1cm)}]
 	\fill[gray!30] (0,0,1) -- (1,0,1) -- (1,1,1) -- (0,1,1) -- cycle;
 	\draw[every edge/.append style={dashed}]
 	(1,0,0) node[below] {$u_i$} -- (0,0,0) node[below left] {$u$} -- (0,0,1) node[left] {$u_k$} -- (0,1,1) node[above] {$u_{jk}$} -- (1,1,1) node[above right] {$u_{ijk}$} -- (1,1,0) node[right] {$u_{ij}$}-- cycle -- (1,0,1) node[right] {$u_{ki}$}-- (0,0,1)
 	(1,0,1) -- (1,1,1)
 	(0,1,0) edge (0,0,0) edge (0,1,1) edge (1,1,0) node[left] {$u_j$};
 	
 	\draw[very thick, red, dashed] (1,1,0) -- (0,1,1) ;
 	\draw[very thick, red] (0,1,1) -- (1,0,1) ;
 	\draw[very thick, red] (1,0,1) -- (1,1,0) ;
 	
 	\draw[very thick, blue] (1,0,0) -- (1,1,0);
 	\draw[very thick, red, dashed] (1,1,0) -- (0,1,0);
 	\draw[very thick, blue, dashed] (0,1,0) -- (0,1,1);
 	\draw[very thick, red] (0,1,1) -- (0,0,1);
 	\draw[very thick, blue] (0,0,1) -- (1,0,1);
 	\draw[very thick, red] (1,0,1) -- (1,0,0);
 	
 	\draw[very thick, blue, dashed] (1,0,0) -- (0,1,0) ;
 	\draw[very thick, blue, dashed] (0,1,0) -- (0,0,1) ;
 	\draw[very thick, blue] (0,0,1) -- (1,0,0) ;
 \node at (1,0,0) {$\bullet$};
 \node at (0,1,0) {$\bullet$};
 \node at (0,0,1) {$\bullet$};
 \node at (1,1,0) {$\bullet$};
 \node at (0,1,1) {$\bullet$};
 \node at (1,0,1) {$\bullet$};
 \node at (1,-1) {(b)};
 \end{scope}
 \end{tikzpicture}
 \caption[.]{(a) The leg structure of a single Lagrangian $\cL_\triang(u_k,u_{ki},u_{jk},\alpha_i,\alpha_j)$. (b) The leg structure for the action on an elementary cube of the triangle 2-form $\cL_\triang$ sits on an octahedral stencil.}
 \label{fig-triang-cube}
 \end{figure}
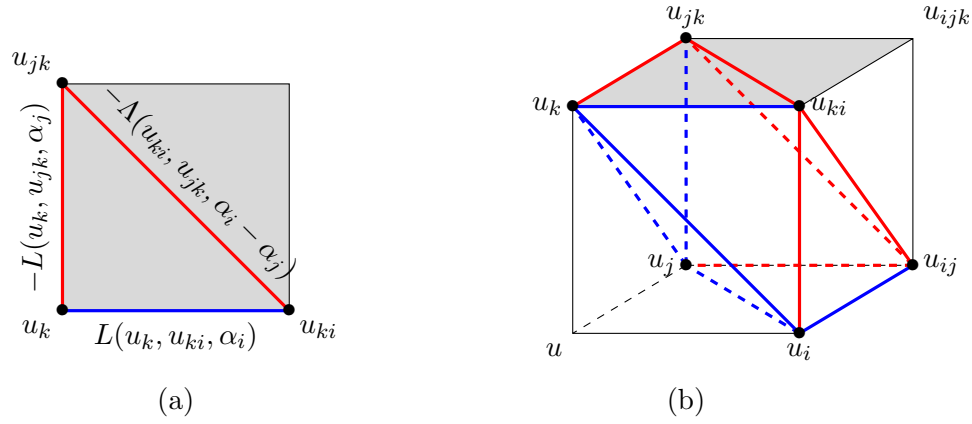

The action around the elementary cube of the triangle Lagrangian depends on an octahedral stencil depicted in Figure~\ref{fig-triang-cube},
\begin{align}\label{Striangle}
 S_\triang &= L(u_i,u_{ij},\alpha_j) + \Lambda(u_i,u_j,\alpha_i-\alpha_j) - \reverse + \cyclic .
\end{align}
Since $u$ and $u_{ijk}$ does not appear in this action, there are no corner equations at these points.
The corner equation at $u_{ij}$ in terms of leg functions is
\begin{align}
	0 &= \pdv{S_\triang}{u_{ij}}
 = \pdv{}{u_{ij}} ( L(u_i,u_{ij},\alpha_j)\hspace{-0.5pt} -\hspace{-0.5pt} L(u_j,u_{ij},\alpha_i)\hspace{-0.5pt} - \hspace{-0.5pt} \Lambda(u_{ij},u_{ki}, \alpha_j-\alpha_k) \hspace{-0.5pt}-\hspace{-0.5pt} \Lambda(u_{jk},u_{ij}, \alpha_k-\alpha_i)) \notag \\
 &= \psi(u_{ij},u_i,\alpha_j) - \psi(u_{ij},u_j,\alpha_i) - \phi(u_{ij},u_{ki}, \alpha_j-\alpha_k) - \phi(u_{ij},u_{jk}, \alpha_k-\alpha_i) \notag \\
 & =: \mathcal{E}_{ij} .\label{triangle-corners}
\end{align}
Note that this corner equation can be written as a combination of two quad equations in three-leg form, based at the vertex $u_{ij}$:
\begin{equation}\label{fourlegE-QQ}
 \mathcal{E}_{ij} = -\mathcal Q^{(u_{ij})}_{jk} - \mathcal Q^{(u_{ij})}_{ki} .
\end{equation}
At the vertex $u_i$, we find the corner equation \smash{$\inversion{\mathcal{E}}_{ij} = 0$}, where \smash{$\inversion{\mathcal{E}}_{ij} := -\mathcal{Q}^{(u_k)}_{jk} - \mathcal{Q}^{(u_k)}_{ki}$} is obtained from $\mathcal{E}_{ij}$ by point inversion in the cube, i.e., $u \leftrightarrow u_{ijk}$, $u_i \leftrightarrow u_{jk}$, etc.

\begin{Proposition}
 \label{prop-triang}
 For ${\rm H}2$, ${\rm H}3$, ${\rm A}1_{\delta\neq0}$, ${\rm A}2$, ${\rm Q}1_{\delta\neq0}$, ${\rm Q}2$, ${\rm Q}3$, the equations $\mathcal{E}_{ij} = 0$, $\mathcal{E}_{jk} = 0$, \dots, $\inversion{\mathcal{E}}_{ki} = 0$ are satisfied if and only if there exist $\Theta,\Theta_i,\dots,\Theta_{ijk} \in \mathbb{Z}$ and $\Xi_i,\Xi_j,\Xi_k \in \mathbb{Z}$ such that
 \begin{alignat}{4}
 &\pdv{S^{\vec \Theta,\vec \Xi}_\triang}{u_i} = 0 , \qquad &&
 \pdv{S^{\vec \Theta,\vec \Xi}_\triang}{u_j} = 0 , \qquad &&
 \pdv{S^{\vec \Theta,\vec \Xi}_\triang}{u_k} = 0 , &\nonumber\\
 &\pdv{S^{\vec \Theta,\vec \Xi}_\triang}{u_{jk}} = 0 , \qquad &&
 \pdv{S^{\vec \Theta,\vec \Xi}_\triang}{u_{ki}} = 0 , \qquad &&
 \pdv{S^{\vec \Theta,\vec \Xi}_\triang}{u_{ij}} = 0 , & \label{triang-corners=0}
 \end{alignat}
 and
 \begin{equation}
 \label{triang-alphas}
 \pdv{S^{\vec \Theta,\vec \Xi}_\triang}{\alpha_i} = 0 , \qquad
 \pdv{S^{\vec \Theta,\vec \Xi}_\triang}{\alpha_j} = 0 , \qquad
 \pdv{S^{\vec \Theta,\vec \Xi}_\triang}{\alpha_k} = 0 .
 \end{equation}

 For ${\rm H}1$, ${\rm A}1_{\delta=0}$ and ${\rm Q}1_{\delta=0}$, the quad equations are satisfied if and only if there exist integers $\Xi_i,\Xi_j,\Xi_k \in \mathbb{Z}$ such that equations~\eqref{triang-corners=0} and~\eqref{triang-alphas} hold for $S^{\vec \Xi}$.
\end{Proposition}
\begin{proof}
 Expressions~\eqref{triang-corners=0} follow from equation~\eqref{triangle-corners} and its analogues obtained by cyclicity and point inversion.

 To derive equations~\eqref{triang-alphas}, we first show that any solution $(u_i,u_j,u_k,u_{ij},u_{jk},u_{ki})$ to the $\mathcal{E}$-equations can be extended to a solution to the quad equations. We have that
 \begin{align*}
 & \inversion{\mathcal{E}}_{ij} = -\mathcal{Q}^{(u_k)}_{jk} - \mathcal{Q}^{(u_k)}_{ki} \equiv 0 \mod 2 \pi {\rm i} , \\
 & \inversion{\mathcal{E}}_{jk} = -\mathcal{Q}^{(u_i)}_{ki} - \mathcal{Q}^{(u_i)}_{ij} \equiv 0 \mod 2 \pi {\rm i} , \\
 & \inversion{\mathcal{E}}_{ki} = -\mathcal{Q}^{(u_j)}_{ij} - \mathcal{Q}^{(u_j)}_{jk} \equiv 0 \mod 2 \pi {\rm i} .
 \end{align*}
 Choose $u$ such that the quad equation $Q_{ij} = 0$ holds. Then \smash{$\mathcal{Q}^{(u_i)}_{ij} \equiv \mathcal{Q}^{(u_j)}_{ij} \equiv 0 \mod 2 \pi {\rm i}$}, hence \smash{$\mathcal{Q}^{(u_i)}_{ki} \equiv \mathcal{Q}^{(u_j)}_{jk} \equiv 0 \mod 2 \pi {\rm i}$}, so the quad equations $Q_{jk} = 0$ and $Q_{ki} = 0$ hold as well. Similarly,~we can choose $u_{ijk}$ such that the three shifted quad equations \smash{$\inversion{Q}_{ij} = 0$}, \smash{$\inversion{Q}_{jk} = 0$}, and \smash{$\inversion{Q}_{ki} = 0$} hold.

\allowdisplaybreaks Now, starting from the definition~\eqref{Striangle} and using equations~\eqref{L-log(h)}--\eqref{Lambda-log(g)}, we compute
 \begin{align*}
	\pdv{S^{\Theta,\Xi}_\triang}{\alpha_i} &=
	\pdv{}{\alpha_i} ( L(u_k,u_{ki},\alpha_i) - L(u_j,u_{ij},\alpha_i) + \Lambda(u_i,u_{j},\alpha_i-\alpha_j) + \Lambda(u_k,u_{i},\alpha_k-\alpha_i)\\
	&\hphantom{=\pdv{}{\alpha_i} ( }{}\, - \Lambda(u_{ki}, u_{jk}, \alpha_i-\alpha_j) - \Lambda(u_{jk},u_{ij},\alpha_k-\alpha_i ) + 2 \Xi_i \pi {\rm i} \\
	&\equiv \log\biggl( \frac{h(u_k,u_{ki},\alpha_i)g(u_i,u_{j},\alpha_i-\alpha_j)g(u_{jk},u_{ij},\alpha_k-\alpha_i)}{h(u_j,u_{ij},\alpha_i)g(u_k,u_{i},\alpha_k-\alpha_i)g(u_{ki}, u_{jk}, \alpha_i-\alpha_j)} \biggr) \mod 2 \pi {\rm i} \\
	&\equiv \log\biggl( \frac{g(u, u_{ij}, \alpha_i-\alpha_j) g(u_i,u_j,\alpha_i-\alpha_j)}{h(u,u_{i},\alpha_i)h(u_j,u_{ij},\alpha_i)} \biggr) \\
	& \quad + \log\biggl( \frac{h(u,u_{i},\alpha_i)h(u_k,u_{ki},\alpha_i)}{g(u, u_{ki}, \alpha_k-\alpha_i)g(u_k,u_i,\alpha_k-\alpha_i)} \biggr) \\
	& \quad + \log\biggl( \frac{g(u_{jk},u_{ij},\alpha_k-\alpha_i)g(u,u_{ki},\alpha_k-\alpha_i)}{g(u_{ik},u_{jk},\alpha_i-\alpha_j)g(u,u_{ij},\alpha_i-\alpha_j) } \biggr) \mod 2 \pi {\rm i} .
\end{align*}
Each of the logarithms vanishes on the quad equations, so the existence of a suitable $\Xi_i$ fol\=lows.\looseness=1
\end{proof}

The triangle Lagrangian 2-form produces a set of corner equations that vanish on the quad equations but are not equivalent to them, because they lack the variables $u$ and $u_{ijk}$. Furthermore, we can relate each of the four-leg equations~\eqref{triangle-corners} to a quad equation and a tetrahedron equation. From $\cL_\triang = \cL_\trileg - \cL_\cross$, we have
\begin{equation}
\label{Striang-tridcross}
S_\triang = S_\trileg - S_\cross
\end{equation}
and thus
\begin{equation}\label{fourlegE-QT}
 \mathcal{E}_{ij} = -\mathcal{Q}_{ij}^{(u_{ij})} - \mathcal{T}^{(u_{ij})} .
\end{equation}

\begin{Example} For H1, also known as the lattice potential KdV equation, we have
\begin{equation}
	\label{H1-Q}
	Q_{ij} = (u_i - u_j)(u - u_{ij}) - \alpha_i + \alpha_j .
\end{equation}
Its leg functions are given by
\[ \psi(u,u_i,\alpha_i) = u_i \qquad \text{and} \qquad \phi(u,u_{ij},\alpha_i-\alpha_j) = \frac{\alpha_i-\alpha_j}{u - u_{ij}} \]
and triangle Lagrangian is
\[ \cL_\triang = u u_i - u u_j - (\alpha_i - \alpha_j) \log( u_i - u_j ) , \]
so
\begin{align*}
 S_\triang & = u_k u_{ki} - u_k u_{jk} - (\alpha_i - \alpha_j) \log( u_{ki} - u_{jk} ) + (\alpha_i - \alpha_j) \log( u_i - u_j )^2 + \cyclic , \\
 & = u_i u_{ij} + (\alpha_i - \alpha_j) \log( u_i - u_j ) - \reverse + \cyclic.
\end{align*}
Two corner equations are identically zero and the other six have a four-leg form. The corner equations at $u_{ij}$, $u_{jk}$ and $u_{ki}$ are all of the form
\begin{align*}
 0 = \pdv{S_\triang}{u_{ij}} &= u_i - u_j - \frac{\alpha_j - \alpha_k}{u_{ij} - u_{ki}} + \frac{\alpha_k - \alpha_i}{u_{jk} - u_{ij}} = -\mathcal Q^{(u_{ij})}_{jk} - \mathcal Q^{(u_{ij})}_{ki} = \mathcal{E}_{ij} .
\end{align*}
Clearly these equations vanish on the quad equations. Similarly,
the other three corner equations are of the form
\begin{align*}
 \pdv{S_\triang}{u_k} &= u_{ki} - u_{jk} + \frac{\alpha_k - \alpha_i}{u_k - u_i} - \frac{\alpha_j - \alpha_k}{u_j - u_k} = \mathcal{Q}^{(u_{k})}_{jk} + \mathcal{Q}^{(u_{k})}_{ki} = -\inversion{\mathcal{E}}_{ij} .
\end{align*}

Furthermore, we have
\begin{align*}
 0 = \pdv{S_\triang^{\vec \Xi}}{\alpha_i} &= \log(u_i-u_j) - \log(u_{jk}-u_{ki}) - \log(u_k-u_i) + \log(u_{ij}-u_{jk}) + 2 \Xi_i \pi {\rm i}\\
 &\equiv \log \biggl( \frac{(u_i-u_j)(u_{ij}-u_{jk})}{(u_k-u_i)(u_{jk}-u_{ki})} \biggr) \mod 2 \pi {\rm i} .
\end{align*}
The only way the principal branch of $\log$ can yield a multiple of $2 \pi {\rm i}$ is for it to be exactly zero. Hence,
\[
\pdv{S_\triang^{\vec \Xi}}{\alpha_i} = 0
\]
 is equivalent to the lattice KP equation~\cite{nijhoff1984backlund}
\begin{align*}
 0 &= (u_i-u_j)(u_{ij}-u_{jk}) - (u_k-u_i)(u_{jk}-u_{ki}) \\
 &= u_i(u_{ij}-u_{ki}) + u_j(u_{jk}-u_{ij}) + u_k(u_{jk}-u_{ki}) .
\end{align*}
In other words, Proposition~\ref{prop-triang} gives us a short proof that the $\mathcal{E}$-equations for H1 imply the discrete KP equation. We will return to this when we discuss octahedron equations in Section~\ref{sec-octa}.
\end{Example}

\subsection{Closure relation}
%\label{subsec-clos}

Propositions~\ref{prop-trileg}--\ref{prop-triang} above provide the basis to prove closure of each of the Lagrangian multiforms, following the same approach we took in~\cite{richardson2025discrete1} for the trident Lagrangian. Propositions~\mbox{\ref{prop-trileg}--\ref{prop-triang}} state that on solutions of the corner equations, the gradient of the action with respect to both parameters $\alpha$ and fields $u$ vanishes. This means that perturbations of a given solution will have the same value of the action, unless it crosses a branch cut in one of the functions making up the action. For H1, H2, Q1, Q2, A1, the effect of crossing a branch cut is cancelled by the change in $\vec \Theta$ (if present) and $\vec \Xi$ that happens at the same point. For~H3,~Q3,~A2, the net effect of crossing a branch cut is that the value of the action changes by a multiple of~$4 \pi^2$.

Thus we arrive at the following statement (see~\cite[Theorem~3.5]{richardson2025discrete1}).

\begin{Theorem}
 \label{thm-closed}
 Let $S$ denote one of the actions \smash{$S_\trileg^{\vec \Theta, \vec \Xi}$}, \smash{$S_\triang^{\vec \Theta, \vec \Xi}$}, \smash{$S_\crosssquare^{\vec \Theta, \vec \Xi}$}, or \smash{$S_\cross^{\vec \Theta, \vec \Xi}$} for ${\rm H}2$, ${\rm H}3$, ${\rm A}1_{\delta\neq0}$, ${\rm A}2$, ${\rm Q}1_{\delta\neq0}$, ${\rm Q}2$, ${\rm Q}3$. Or, let $S$ denote one of the actions \smash{$S_\trileg^{\vec \Xi}$}, \smash{$S_\triang^{\vec \Xi}$}, \smash{$S_\crosssquare^{\vec \Xi}$}, or \smash{$S_\cross^{\vec \Xi}$} for ${\rm H}1$, ${\rm A}1_{\delta=0}$, ${\rm Q}1_{\delta=0}$.

 For any solution $\vec u$ to the corner equations of $S$, with associated integer fields $\vec \Theta$, $\vec \Xi$, there holds
 \[ S(\vec u, \vec \Theta, \vec \Xi) =
 \begin{cases}
 0 & \text{for ${\rm H}1$, ${\rm H}2$, ${\rm Q}1$, ${\rm Q}2$, and ${\rm A}1$,} \\
 4 k \pi^2,\ k \in \Z, & \text{for ${\rm H}3$, ${\rm Q}3$, and ${\rm A}2$.}
 \end{cases}\]
\end{Theorem}

\section{Tetrahedron and octahedron equations}
\label{sec-octa}

Previously, we discussed the three-leg forms equivalent to the multi-affine quad and tetrahedron equations. In this section, we discuss the polynomial equations associated with the four-leg corner equations~\eqref{triangle-corners} of the triangle 2-form. We review the related multi-affine \emph{octahedron equations} (also known as \emph{octahedron relations}) and their relationship to quad equations. Most of this section is based on~\cite{boll2016integrability}, but in Proposition~\ref{prop-484} we provide an additional variational interpretation of the relation between the different types of polynomial equations.

The four-leg equation $\mathcal E_{ij} = 0$ or $\mathcal E_{ij} \equiv 0 \mod 2 \pi {\rm i}$, where the left-hand side is defined in~equation~\eqref{fourlegE-QQ}, can be obtained by eliminating the variable $u_{ijk}$ from two quad equations in~three-leg form. An equivalent equation is obtained by eliminating $u_{ijk}$ from the corresponding two multi-affine quad equations. This yields the equation $E_{ij} = 0$, where the left-hand side is the polynomial
\begin{equation}\label{E-Q}
 E_{ij} = \pdv{\inversion Q_{jk}}{u_{ijk}} \inversion{Q}_{ki} - \pdv{\inversion{Q}_{ki}}{u_{ijk}} \inversion{Q}_{jk} ,
\end{equation}
where \smash{$\inversion{\cdot}$} denotes point inversion in the cube, for example \smash{$\inversion Q_{jk} = Q(u_{ijk},u_{ki},u_{ij},u_{i},\alpha_j-\alpha_k)$}.
We have two more polynomials by permuting indices and three more from point inversion:
\begin{equation}
\label{E-Q-invert}
\inversion{E}_{ij} = \pdv{Q_{jk}}{u} Q_{ki} - \pdv{Q_{ki}}{u} Q_{jk} .
\end{equation}
These polynomials lead to equations that are equivalent to those involving the four-leg expressions $\mathcal{E}_{ij}$:
\[ E_{ij}(u_i,u_j,u_{ij},u_{jk},u_{ki},\alpha_i,\alpha_j,\alpha_k) = 0\]
is equivalent to
\[ \mathcal{E}_{ij}(u_i,u_j,u_{ij},u_{jk},u_{ki},\alpha_i,\alpha_j,\alpha_k) = 0 \]
for H1, Q1$_{\delta=0}$, A1$_{\delta=0}$, and to
\[ \mathcal{E}_{ij}(u_i,u_j,u_{ij},u_{jk},u_{ki},\alpha_i,\alpha_j,\alpha_k) \equiv 0 \mod 2 \pi {\rm i} \]
for H2, H3, A1$_{\delta\neq0}$, A2, Q1$_{\delta\neq0}$, Q2, Q3.

From their definition as combinations of quad equations, it is clear that $E_{ij}=0$ and $\inversion E_{ij}=0$ are consequences of the quad equations. However, they are not equivalent to the quad equations. Only two of these six equations are independent (see Proposition~\ref{prop-2outof8} below).

The relation~\eqref{fourlegE-QT} suggests that the polynomial $E_{ij}$ is related to the polynomials $Q_{ij}$ and~$T$. Indeed, eliminating the variable $u$ from the system $T=0$, $Q_{ij}=0$ must lead to an equation equivalent to $E_{ij} = 0$, so we find
\begin{equation}\label{TQE-rel}
 \pdv{Q_{ij}}{u} T - \pdv{T}{u} Q_{ij} = \gamma(\alpha_i,\alpha_j,\alpha_k) E_{ij} .
\end{equation}
Here, $\gamma(\alpha_i,\alpha_j,\alpha_k)$ does not contain field variables, because the left hand side depends linearly on $(u_i,u_j,u_{jk},u_{ki})$ and quadratically on $u_{ij}$, as does $E_{ij}$.
Under point inversion, we find the analogous relation
\begin{equation*} \pdv{\inversion Q_{ij}}{u_{ijk}} \inversion T - \pdv{\inversion T}{u_{ijk}} \inversion Q_{ij} = \gamma(\alpha_i,\alpha_j,\alpha_k) \inversion E_{ij} .
\end{equation*}

The main result of~\cite{boll2016integrability} is that for every member of the ABS list there exist two octahedron equations which are equivalent to the set of $E$-equations. These are of the form
\begin{align*}
 \Omega_1(u_i,u_j,u_k,u_{ij}, u_{jk},u_{ki}, \alpha_i,\alpha_j,\alpha_k) = 0 , \\
 \Omega_2(u_i,u_j,u_k,u_{ij}, u_{jk},u_{ki},\alpha_i,\alpha_j,\alpha_k) = 0 ,
\end{align*}
where $\Omega_1$, $\Omega_2$ are multi-affine polynomials in the field variables $u_i$, $u_j$, $u_k$, $u_{ij}$, $u_{jk}$, $u_{ki}$, which form an octahedral stencil. These polynomials are (anti-)symmetric under point inversion:
For~Q4 and~A2, we have \smash{$\inversion \Omega_1 = \Omega_1$} and \smash{$\inversion \Omega_2 = \Omega_2$}; for the rest of ABS list, there holds \smash{$\inversion \Omega_1 = -\Omega_1$} and \smash{$\inversion \Omega_2 = \Omega_2$}.

In all cases, we can get $E$-polynomials by eliminating a variable form the system of octahedron equations. For Q4 and A2, eliminating $u_{k}$ and $u_{ij}$ leads to, respectively,
\begin{align}
 &\pdv{\Omega_1}{u_k} \Omega_2 - \pdv{\Omega_2}{u_k} \Omega_1 = \mu(u_i,u_j,u_{jk},u_{ki},\alpha_i,\alpha_j,\alpha_k) E_{ij} , \nonumber \\
 &\pdv{\Omega_1}{u_{ij}} \Omega_2 - \pdv{\Omega_2}{u_{ij}} \Omega_1 = \mu(u_{jk},u_{ki},u_i,u_j,\alpha_i,\alpha_j,\alpha_k) \inversion E_{ij} ,\label{Omegas-mu-E}
\end{align}
where the factor $\mu(u_i,u_j,u_{jk},u_{ki},\alpha_i,\alpha_j,\alpha_k)$ is polynomial in $u_i$, $u_j$, $u_{jk}$, $u_{ki}$.
For the rest of~ABS list, we have the explicit expressions
\begin{align}
 &\pdv{\Omega_1}{u_k} \Omega_2 - \pdv{\Omega_2}{u_k} \Omega_1 = g_1(\alpha_i,\alpha_j,\alpha_k) \biggl( \pdv{\Omega_1}{u_{ij}} - \pdv{\Omega_1}{u_k} \biggr) E_{ij} , \nonumber\\
 &-\pdv{\Omega_1}{u_{ij}} \Omega_2 + \pdv{\Omega_2}{u_{ij}} \Omega_1 = g_1(\alpha_i,\alpha_j,\alpha_k) \biggl( -\pdv{\Omega_1}{u_{k}} + \pdv{\Omega_1}{u_{ij}} \biggr) \inversion E_{ij} ,\label{Omegas-explicit-E}
\end{align}
where $g_1$ is defined explicitly on a case by case basis in~\cite[Proposition 5.3]{boll2016integrability}.

\begin{Example}[H1 octahedron equations]
 For H1, the four-leg equations $\mathcal{E}_{ij} = 0$ and $\inversion{\mathcal{E}}_{ij} = 0$ are equivalent to the following polynomials, respectively:
 \begin{align*}
 &E_{ij} = (u_i - u_j) (u_{ij} - u_{jk}) (u_{ij} - u_{ki} ) + \alpha_i (u_{ki} - u_{ij} ) + \alpha_j (u_{ij} - u_{jk} ) + \alpha_k ( u_{jk} - u_{ki} ) , \\
 &\inversion{E}_{ij} = (u_{jk} - u_{ki} ) (u_k - u_i) (u_k - u_j) + \alpha_i ( u_j - u_k ) + \alpha_j (u_k - u_i ) + \alpha_k ( u_i - u_j ) .
 \end{align*}
 These can written in terms of the following multi-affine octahedron polynomials:
 \begin{align}
 &\Omega_1= u_i ( u_{ki} - u_{ij} ) + \cyclic , \label{H1-Omega1} \\
 &\Omega_2 = \alpha_i ( u_k - u_j + u_{ij} - u_{ki} ) + u_i u_{jk} ( u_j - u_k - u_{ij} u_{ki} ) + \cyclic ,\nonumber
 \end{align}

 Equation~\eqref{H1-Omega1} is exactly the discrete KP equation we obtained earlier from \smash{$\pdv{S_\triang}{\alpha_i} = 0$}. The latter can be understood as a conservation law,
 \[ 0 = \pdv{S}{\alpha_i} = \Delta_k \pdv{\cL_{ij}}{\alpha_i} + \Delta_j \pdv{\cL_{ki}}{\alpha_i} , \]
 where $\Delta_k$ denotes the difference between the shift in the $k$-direction and the function itself.
 This is the 2-form version of the \emph{spectrality property} that was formulated in~\cite{suris2013variational} for discrete Lagrangian 1-forms.
 Further investigation is required to determine how this property relates to the octahedron equations of~\cite{boll2016integrability} for other members of the ABS list.
\end{Example}

The equivalence between the corner equations of $\cL_\triang$ (i.e., the $E$-equations obtained by setting the expressions~\eqref{E-Q} and~\eqref{E-Q-invert} to zero) and octahedron equations
can be seen as a particular case of the following statement, which holds in the sense of fractional ideals, as explained in~\cite{boll2016integrability}.

\begin{Proposition}
\label{prop-2outof8}
 Out of the set of equations consisting of the six $E$-equations and two octahedron equations, any $2$ equations imply the other six.
\end{Proposition}

Further to this, a dimension-counting argument suggests that if we add the tetrahedron equations to the set of octahedron equations (or $E$-equations), we should obtain a system equivalent to the quad equations. The Lagrangian multiform structure provides an explicit proof of this~fact:

\begin{Proposition}
 \label{prop-484}
 The tetrahedron equations and octahedron equations together are equivalent to the quad equations.
\end{Proposition}

\begin{proof}
 From equation~\eqref{tetra-u} (or from tetrahedron property as assumed in the ABS classification~\cite{adler2003classification}), it follows that the quad equations imply the tetrahedron equations. From~\eqref{Omegas-mu-E} or~\eqref{Omegas-explicit-E}, it follows that the $E$-equations, and hence the quad equations, imply the octahedron equations (in the sense of fraction ideals, as in~\cite{boll2016integrability}).

 To prove the other implication, assume that the octahedron and tetrahedron equations are satisfied. Then the actions of $\cL_\triang$ and $\cL_\cross$ are critical. Since
 $\cL_\trileg = \cL_\triang + \cL_\cross$, it follows that the action of $\cL_\trileg$ is critical, hence the quad equations hold.
\end{proof}

\begin{table}[ht]
\centering
\renewcommand{\arraystretch}{1.5}
\begin{tabular}{ |c|c|c| }
\hline
2-Form & Corner polynomials & Equivalent system \\
\hline
$\cL_\triang$ & $\bigl(0, \inversion{E}_{jk}, \inversion{E}_{ki}, \inversion{E}_{ij}, E_{ij}, E_{jk}, E_{ki}, 0\bigr)$ & $\Omega_1,\Omega_2 = 0$ \\
$\cL_\cross$ & $\bigl(T, \inversion{T}, \inversion{T}, \inversion{T}, T, T, T, \inversion{T}\bigr)$ & $T, \inversion{T} = 0$ \\
$\cL_\trileg$ & $\bigl(T, \inversion{Q}_{jk}, \inversion{Q}_{ki}, \inversion{Q}_{ij}, Q_{ij}, Q_{jk}, Q_{ki}, \inversion{T}\bigr)$ & $\Omega_1, \Omega_2, T, \inversion{T} = 0$ \\
$\cL_\crosssquare$ & $\bigl(T, \sim \inversion{Q}_{jk}, \sim \inversion{Q}_{ki}, \sim \inversion{Q}_{ij}, \sim Q_{ij}, \sim Q_{jk}, \sim Q_{ki}, \inversion{T}\bigr)$ & $\Omega_1, \Omega_2, T, \inversion{T} = 0$ \\
\hline
\end{tabular}
\caption[.]{Overview of the four types of Lagrangian multiform, with their corner equations (in polynomial form) and a symmetric set of equations forming an equivalent system. Note in the final row ``\smash{$\sim \inversion{Q}_{jk}$}'' represents an expression such that, after the elimination of a tetrahedron polynomial \smash{$\inversion T$}, a quad polynomial $\inversion{Q}_{jk}$ remains.}
\label{table-corner-eqns}
\end{table}

Proposition~\ref{prop-484} relates the three sets of equations we are dealing with: those produced by~$\cL_\trileg$ (or $\cL_\crosssquare$), $\cL_\cross$, and $\cL_\triang$. (See Table~\ref{table-corner-eqns} for an overview.)
The corner equations of $\cL_\cross$ can be immediately identified with the tetrahedron equations.
The equations produced by $\cL_\triang$ can be understood from two points of view. On the one hand, they are generated by the 5-point $E$-equations, which have an obvious variational interpretation, but are far less symmetric than the other equations considered. On the other hand they are generated by the octahedron equations, which have cyclic and point inversion symmetry, but have no previously known variational interpretation.
The equations produced by $\cL_\trileg$ (or $\cL_\crosssquare$) are the six quad equations around the cube (together with the two tetrahedron equations), which are equivalent to the combined set of two tetrahedron equations and two octahedron equations.

In the next section, we introduce the \emph{double zero property}, which states that the exterior derivative of a 2-form can be expanded such that each of its terms is a product of expressions that vanish on its corner equations. This will give us a variational interpretation for all of the equations discussed above, including the octahedron equations.

\section{The double zero property of discrete exterior derivatives}
\label{sec-double0}

To motivate the formulation of the double zero property, we first present some observations regarding a linear quad equation. Then we formalise the definition of a double zero expansion and provide constructions of double zero expansions for each of the 2-forms for any member of ABS list. Afterwards we investigate explicit expansions associated with each of the 2-forms for~H1.

\subsection{Double zero expansion associated with a linear quad equation}

We consider the linear quad equation $Q_{ij} = 0$, with
\begin{align*}
 Q_{ij} &= (\alpha_i + \alpha_j)(u_i - u_j) - (\alpha_i - \alpha_j)(u - u_{ij}) .
\end{align*}
This equation can be considered as a linearisation of H1 and is associated to the following triangle Lagrangian
\begin{align*}
 \cL_\triang & = u (u_i - u_j) - \frac{\alpha_i + \alpha_j}{2(\alpha_i - \alpha_j)} (u - u_{ij})^2 .
\end{align*}
Its action over an elementary cube can be written as
\begin{align}
 S_\triang & = \cL_\triang(u_k,u_{ki},u_{jk},\alpha_i,\alpha_j) + \cL_\triang(u_i,u_{ij},u_{ki},\alpha_j,\alpha_k) + \cL_\triang(u_j,u_{jk},u_{ij},\alpha_k,\alpha_i) \nonumber \\
 & \quad - \cL_\triang(u,u_i,u_j,\alpha_i,\alpha_j) - \cL_\triang(u,u_j,u_k,\alpha_j,\alpha_k) - \cL_\triang(u,u_k,u_i,\alpha_k,\alpha_i) \nonumber \\
 & = \frac{O_1 \, O_2}{(\alpha_i - \alpha_j)(\alpha_j - \alpha_k)(\alpha_k - \alpha_i)} , \label{eqn:lindoubzeroo1o2}
\end{align}
where
\begin{align*}
 O_1 & := (\alpha_j - \alpha_k) u_i + (\alpha_j - \alpha_k) u_{jk} + \cyclic , \\
 O_2 & := \alpha_i (\alpha_j - \alpha_k) u_i - \alpha_i (\alpha_j - \alpha_k) u_{jk} + \cyclic .
\end{align*}
Hence, the action over an elementary cube can be explicitly written as a product of two expressions $O_1$ and $O_2$. They are symmetric under cyclic permutation of the indices and under point inversion, in the sense that \smash{$\inversion{O}_1 = O_1$} and \smash{$\inversion{O}_2 = -O_2$}.

In analogy to multiple zeros of polynomials, we say that $S_\triang$ has a \emph{double zero} on the system of equations $O_1=0$, $O_2=0$.
The factorisation~\eqref{eqn:lindoubzeroo1o2} implies that for any $v \in \{u,u_i,u_j,u_k,u_{ij},u_{jk},\allowbreak u_{ki},u_{ijk}\}$ there holds
\[ \pdv{ S_\triang }{v} = \frac{1}{(\alpha_i - \alpha_j)(\alpha_j - \alpha_k)(\alpha_k - \alpha_i)} \biggl( \pdv{O_1}{v} O_2 + O_1 \pdv{O_2}{v} \biggr), \]
which is zero if both $O_1=0$ and $O_2=0$. Hence, the double zero property tells us that the action is critical if the two equations $O_1=0$ and $O_2=0$ hold. This implies that all corner equations are consequences of these two equations!

This linear example shows the power of the double zero property: if $S$ can be written as a product of two expressions, then the fact that both these expressions vanish is a sufficient condition for criticality. Hence, in this situation, two equations together imply the full system of corner equations.

In some examples, the double zero expansion of $S$ involves more than two expressions and is a sum of products, rather than a single product. In the next subsection, we give a suitably general definition of a double zero expansion.

\subsection{Definition of a double zero expansion}

\begin{Definition}\label{def:doublezero}
 We say that $S(u,u_i,u_j,u_k,u_{ij},u_{jk},u_{ki},u_{ijk})$ has a \emph{double zero} on a set of equations $\{ K_m(u,u_i,u_j,u_k,u_{ij},u_{jk},u_{ki},u_{ijk})=0 \mid m = 1, \dots, M \}$ if it can be written as
 \begin{equation}
 \label{double0def}
 S = \sum_{1 \leq m \leq m' \leq M} d_{m,m'} K_m \, K_{m'} .
 \end{equation}
 Here, $d_{m,m'}$ are coefficients that can depend upon any variable or parameter, but must be nonsingular on generic points of $\{ K_m(u,u_i,u_j,u_k,u_{ij},u_{jk},u_{ki},u_{ijk})=0 \mid m = 1, \dots, M \}$.
 We call the right hand side of equation~\eqref{double0def} a \emph{double zero expansion} of $S$.
\end{Definition}

This definition applies to the example above, because equation~\eqref{eqn:lindoubzeroo1o2} is of the form~\eqref{double0def} with~${M=2}$, $K_1 = O_1$, $K_2 = O_2$, $d_{1,1} = d_{2,2} = 0$, and
\[d_{1,2} = \frac{1}{(\alpha_i - \alpha_j)(\alpha_j - \alpha_k)(\alpha_k - \alpha_i)} .\]

\begin{Proposition}
%\label{prop:double0}
 If the action over an elementary cube $S$ of a Lagrangian $2$-form $L$ has a~double zero on a~system of equations $\{ K_m=0 \mid m = 1, \dots, M\}$, then this system implies the corner equations of $L$.
\end{Proposition}
\begin{proof}
 For any $v \in \{u,u_i,u_j,u_k,u_{ij},u_{jk},u_{ki},u_{ijk}\}$, we have
 \[
 \pdv{S}{v} = \sum_{1 \leq m \leq m' \leq M} \biggl( \pdv{d_{m,m'}}{v} K_m K_{m'} + d_{m,m'} \pdv{K_m}{v} K_{m'} + d_{m,m'} K_m \pdv{K_{m'}}{v} \biggr),
 \]
 which vanishes on $\{ K_m=0 \mid m = 1, \dots, M \}$.
\end{proof}

Below, we will use Taylor expansions to derive double zero expansions involving higher powers of expressions. For example, consider an action $S$ which is critical on some equation $K_1 = 0$, and an expansion of the form
\[ S = \sum_{n = 2}^\infty c_{n} K_1^n . \]
This satisfies Definition~\ref{def:doublezero}, because we can write the action as
\[ S = \Biggl( \sum_{n = 2}^\infty c_{n} K_1^{n-1} \Biggr) K_1 , \]
where the sum can be identified with $d_{1,1}$.

\subsection{Double zero expansions for ABS 2-forms}
\label{sec-double0-ABS}

In this subsection, we derive double zero expansions for the action over an elementary cube for each of the discrete 2-forms for arbitrary members of the ABS list. For ease of notation, we assume that we are in a case where the closure property holds:

\begin{Assumption}
\label{assumption}
In the following we assume that $S^{\vec \Theta, \vec \Xi} = 0$ $\bigl($or $S^{\vec \Xi} = 0$ for ${\rm H}1$, ${\rm A}1_{\delta=0}$, ${\rm Q}1_{\delta=0}\bigr)$ on solutions. To drop this assumption, a multiple of $4 \pi^2$ should be added as constant term to the series expansions that follow.
\end{Assumption}

The double zero expansions are constructed from Taylor expansions in one variable, where the variable represents either a quad polynomial $Q$, a tetrahedron polynomial $T$ or an $E$ polynomial. In this Taylor expansion, the zeroth order term vanishes due to the closure relation and the first order term vanishes as a consequence of the corner equations. Hence, each discrete 2-form has a~double zero expansion in terms of the polynomials associated with its corner equations. These double zero expansions are manifestly symmetric under cyclic permutation of the indices and under point inversion.

{\bf Cross-square 2-form.}
Recall that $Q_{ij} = 0$ implies that $\nabla \cL_\crosssquare = 0$, , where $\nabla$ denotes the gradient with respect to $u$, $u_i$, $u_j$, $u_{ij}$ (see Lemma~\ref{lemma-crosssquare})
and that \smash{$S^{\vec \Theta,\vec \Xi}_\crosssquare = 0$} (see Assumption~\ref{assumption}). To get a more explicit picture of the behaviour of $\cL_\crosssquare$ and $S^{\vec \Theta,\vec \Xi}_\crosssquare = 0$ near solutions, we would like to perform a Taylor expansion about solutions $Q_{ij} = 0$. To this end, it is useful to consider~$Q_{ij}$ as one of the variables on which $\cL_\crosssquare$ depends. Specifically, we consider a change of variables from $(u,u_i,u_j,u_{ij})$ to $(Q_{ij},u_i,u_j,u_{ij})$.
Since the quad polynomial $Q_{ij}$ is multi-affine, we can write it~as
\[ Q_{ij}(u,u_i,u_j,u_{ij},\alpha_i,\alpha_j) = r(u_i,u_j,u_{ij},\alpha_i,\alpha_j) u + s(u_i,u_j,u_{ij},\alpha_i,\alpha_j), \]
so it is possible to write $u$ as a rational function of $Q_{ij},u_i,u_j,u_{ij}$:
\begin{equation*}
	u = V(Q_{ij},u_i,u_j,u_{ij},\alpha_i,\alpha_j) := \frac{Q_{ij} - s(u_i,u_j,u_{ij},\alpha_i,\alpha_j)}{r(u_i,u_j,u_{ij},\alpha_i,\alpha_j)} .
\end{equation*}

In the case of the cross square 2-form, we can consider the Lagrangian on a single square, apply this variable transformation, and then Taylor expand about $Q_{ij} = 0$:
\begin{align*}
	\cL_\crosssquare
	& = \cL_\crosssquare(V(Q_{ij},u_i,u_j,u_{ij},\alpha_i,\alpha_j),u_i,u_j,u_{ij},\alpha_i,\alpha_j) \nonumber \\
	& = \sum_{n=0}^\infty \biggl( \frac{1}{n!} \frac{\partial^n}{\partial {Q_{ij}}^n} \cL_\crosssquare(V(Q_{ij},u_i,u_j,u_{ij},\alpha_i,\alpha_j),u_i,u_j,u_{ij},\alpha_i,\alpha_j) \bigg|_{Q_{ij}=0} \biggr) Q_{ij}^n . \nonumber
\end{align*}
Now, the fact that for $Q_{ij} = 0$ there holds $ \nabla \cL_\crosssquare = 0$ implies that the first order term vanishes. Hence, we have
\begin{align*}
	\cL_\crosssquare &= g(\alpha_i,\alpha_j)
 + \cL_\crosssquare\bigl(V(0,u_i,u_j,u_{ij},\alpha_i,\alpha_j),u_i,u_j,u_{ij},\alpha_i,\alpha_j\bigr) \\
 &\quad{} +\sum_{n=2}^\infty \biggl( \frac{1}{n!} \frac{\partial^n}{\partial {Q_{ij}}^n} \cL_\crosssquare\bigl(V(Q_{ij},u_i,u_j,u_{ij},\alpha_i,\alpha_j),u_i,u_j,u_{ij},\alpha_i,\alpha_j\bigr) \Big|_{Q_{ij}=0} \biggr) Q_{ij}^n .
\end{align*}
Together with the fact that $S^{\vec \Theta,\vec \Xi}_\crosssquare = 0$ on solutions, this implies that the series expansion of the action over the elementary cube has no constant or linear terms in $Q_{ij}$.
\begin{Proposition}
The action of the Lagrangian $2$-form $\cL_\crosssquare$ over an elementary cube has \linebreak a~double~zero expansion in terms of $Q_{ij}$, $Q_{jk}$, $Q_{ki}$, \smash{$\inversion Q_{ij}$}, \smash{$\inversion Q_{jk}$}, \smash{$\inversion Q_{ki}$}:
	\begin{align}
		S_\crosssquare={}& {-}\sum_{n = 2}^\infty \biggl( \frac{1}{n!} \frac{\partial^n}{\partial {Q_{ij}}^n} \cL_\crosssquare(V(Q_{ij},u_i,u_j,u_{ij},\alpha_i,\alpha_j),u_i,u_j,u_{ij},\alpha_i,\alpha_j) \bigg|_{Q_{ij}=0} \biggr) Q_{ij}^n \nonumber \\
 &{}{-} \reverse + \cyclic .\label{Scrosssquare-Q}
 \end{align}
\end{Proposition}

{\bf Cross 2-form.}
Recall that \smash{$T,\inversion T = 0$} implies the existence of $\vec \Theta$ and $\vec \Xi$ such that $S_\cross^{\vec \Theta, \vec \Xi} = 0$ (see Theorem~\ref{thm-closed}) and $\nabla S_\cross^{\vec \Theta, \vec \Xi} = 0$, where $\nabla$ denotes the gradient with respect to $u,\dots,u_{ijk}$ (these are the corner equations, see Proposition~\ref{prop-cross}).
With that in mind, and with the intention of finding a series expansion about \smash{$T,\inversion T = 0$}, we consider a change of variables form $(u,u_i,u_j,u_k,u_{ij},u_{jk},u_{ki},u_{ijk})$ to \smash{$\bigl(T,u_i,u_j,u_k,u_{ij},u_{jk},u_{ki},\inversion T\bigr)$}.
Since the tetrahedron polynomial $T$ is multi-affine, it is possible to rewrite $u$ in terms of $T$ \big(and $u_{ijk}$ in terms of \smash{$\inversion T$}\big) in the following way for the entire ABS list:
\begin{align*}
	&u = W(T,u_{ij},u_{jk},u_{ki},\alpha_i,\alpha_j,\alpha_k) , \qquad
	 u_{ijk} = W \bigl(\inversion T,u_{k},u_{i},u_{j},\alpha_i,\alpha_j,\alpha_k \bigr) .
\end{align*}
Here, $W$ is a fraction with a numerator which is multi-affine in $(T,u_{ij},u_{jk},u_{ki})$ and a denominator which is multi-affine in $(u_{ij},u_{jk},u_{ki})$.

Now we consider the action on an elementary cube of the cross 2-form, apply the variable transformation, and Taylor expand about $T=0$ and \smash{$\inversion T=0$}:
\begin{align*}
S_\cross={}& {-}\cL_\cross(u,u_i,u_j,u_{ij},\alpha_i,\alpha_j) - \reverse + \cyclic \\
	={}& {-}\Lambda(u_i, u_j,\alpha_i-\alpha_j) + \Lambda\bigl(W(T,u_{ij},u_{jk},u_{ki},\alpha_i,\alpha_j\bigr), u_{ij},\alpha_i-\alpha_j) - \reverse + \cyclic \\
	={}& {-}\Lambda(u_i, u_j,\alpha_i-\alpha_j) + \sum_{n=0}^\infty \biggl( \frac{1}{n!} \frac{\partial^n}{\partial {T}^n} \Lambda(W(T,u_{ij},u_{jk},u_{ki},\alpha_i,\alpha_j), u_{ij},\alpha_i-\alpha_j) \bigg|_{T=0} \biggr) T^n \\
 &{} {-} \reverse + \cyclic .
\end{align*}
Now, the fact that \smash{$T,\inversion T = 0 \implies S_\cross^{\vec \Theta, \vec \Xi}$}, $\nabla S_\cross^{\vec \Theta, \vec \Xi} = 0$ means that the zeroth and first order terms must cancel against the terms involving $\vec \Theta$ and $\vec \Xi$, so we have the following.
\begin{Proposition}
	The action of the Lagrangian $2$-form $\cL_\cross$ over an elementary cube has a~double~zero expansion in terms of $T$, \smash{$\inversion T$}
	\begin{equation}\label{Scross-T}
		S_\cross^{\vec \Theta, \vec \Xi} = \sum_{n=2}^\infty \biggl( \frac{1}{n!}\frac{\partial^n}{\partial {T}^n} \Lambda(W, u_{ij},\alpha_i-\alpha_j) \bigg|_{T=0} \biggr) T^n - \reverse + \cyclic ,
	\end{equation}
 where $W = W(T,u_{ij},u_{jk},u_{ki},\alpha_i,\alpha_j)$.
\end{Proposition}

{\bf Triangle 2-form.} From equations~\eqref{Scrosssquare-tridcross} and~\eqref{Striang-tridcross}, we infer that
\begin{equation}
 \label{Striang-crosssquarecross}
 S_\triang = \frac12 S_\crosssquare - \frac12 S_\cross .
\end{equation}
We note that $S_\triang(u_i,u_j,u_k,u_{ij},u_{jk},u_{ki})$ does not depend on $u$ or $u_{ijk}$. Thus, imposing the equations \smash{$T,\inversion T = 0$} has no effect on $S_\triang$: we can take $T = 0$ and \smash{$\inversion T = 0$} as the definition of~$u$ and~$u_{ijk}$ in terms of the variables $(u_i,u_j,u_k,u_{ij},u_{jk},u_{ki})$ that occur in $S_\triang$. Since \smash{$T,\inversion T = 0$} implies $S_\cross = 0$, we have that
\begin{equation*}
	T,\inversion T = 0 \ \implies \ S_\triang = \frac12 S_\crosssquare .
\end{equation*}
Now recall that \smash{$E,\inversion E = 0$} implies $S_\triang^{\vec \Theta, \vec \Xi} = 0$ (see Theorem~\ref{thm-closed}) and $\nabla S_\triang^{\vec \Theta, \vec \Xi} = 0$ (see Proposition~\ref{prop-triang}).
With all of this in mind, we note that the polynomial identity~\eqref{TQE-rel} between $T$, $Q_{ij}$ and $E_{ij}$ implies that
\begin{align*}
	T = 0 \ \implies \ Q_{ij} = \frac{\gamma(\alpha_i,\alpha_j,\alpha_k)}{\pdv{T}{u}} E_{ij} .
\end{align*}
Applying these observations to the double zero expansion of the action of the cross square Lagrangian~\eqref{Scrosssquare-Q}, we find the following.
\begin{Proposition}
	The action of the Lagrangian $2$-form $\cL_\triang$ over an elementary cube has \linebreak a~double~zero expansion in terms of $E_{ij}$, $E_{jk}$, $E_{ki}$, \smash{$\inversion E_{ij}$}, \smash{$\inversion E_{jk}$}, \smash{$\inversion E_{ki}$}
	\begin{equation}\label{Striangle-E}
		S_\triang^{\vec \Theta, \vec \Xi} = -\frac{1}{2} \sum_{n = 2}^\infty \frac{\frac{\partial^n}{\partial {Q_{ij}}^n} \cL_\crosssquare(V,u_i,u_j,u_{ij},\alpha_i,\alpha_j) \big|_{Q_{ij}=0} }{n! \bigl( \pdv{T}{u} \frac{1}{\gamma(\alpha_i,\alpha_j,\alpha_k)} \bigr)^n} E_{ij}^n - \reverse + \cyclic ,
	\end{equation}
 where $V = V(Q_{ij},u_i,u_j,u_{ij},\alpha_i,\alpha_j)$.
\end{Proposition}

{\bf Trident 2-form.} Recall that \smash{$Q,\inversion Q = 0$} implies $S_\trileg^{\vec \Theta, \vec \Xi} = 0$ (see Theorem~\ref{thm-closed}) and $\nabla S_\trileg^{\vec \Theta, \vec \Xi} = 0$ (see Proposition~\ref{prop-trileg}). From equation~\eqref{Scrosssquare-tridcross}, we know that
\begin{equation}
 \label{Strileg-crosscrosssquare}
 S_\trileg = \frac{1}{2} S_\cross + \frac{1}{2} S_\crosssquare .
\end{equation}
From the three-leg tetrahedron equations~\eqref{tetra-u} and~\eqref{tetra-uijk}, we can derive the following relations between the multi-affine polynomials $T$ in terms of $Q_{ij}$, $Q_{jk}$ and $Q_{ki}$:
 \begin{align}
 &c T = d_{ij} Q_{ij} + d_{jk} Q_{jk} + d_{ki} Q_{ki} , \nonumber\\
 &\inversion c \, \inversion T= \inversion d_{ij} \inversion Q_{ij} + \inversion d_{jk} \inversion Q_{jk} + \inversion d_{ki} \inversion Q_{ki} ,\label{T-cyclic-Qs}
 \end{align}
where $d$ and $c_{ij}$ are polynomials. We can apply these observations to the double zero expansion of the cross 2-form~\eqref{Scross-T} and cross-square 2-form~\eqref{Scrosssquare-Q} and conclude the following.

\begin{Proposition}
	The action of Lagrangian $2$-form $\cL_\trileg$ over a an elementary cube has \linebreak a~double~zero expansion in terms of $Q_{ij}$, $Q_{jk}$, $Q_{ki}$, \smash{$\inversion Q_{ij}$}, \smash{$\inversion Q_{jk}$}, \smash{$\inversion Q_{ki}$}
	\begin{align*}
 S_\trileg^{\vec \Theta, \vec \Xi}&=\frac{1}{2} \sum_{n=2}^\infty \biggl( \frac{1}{n!}\frac{\partial^n}{\partial {T}^n} \Lambda(W, u_{ij},\alpha_i-\alpha_j ) \bigg|_{T=0} \biggr) \frac{1}{c^n} (d_{ij} Q_{ij} + d_{jk} Q_{jk} + d_{ki} Q_{ki} )^n \\
 &\quad{}-\frac{1}{2} \sum_{n = 2}^\infty \biggl( \frac{1}{n!} \frac{\partial^n}{\partial {Q_{ij}}^n} \cL_\crosssquare (V,u_i,u_j,u_{ij},\alpha_i,\alpha_j ) \bigg|_{Q_{ij}=0} \biggr) Q_{ij}^n - \reverse + \cyclic ,
	\end{align*}
 where $W = W(T,u_{ij},u_{jk},u_{ki},\alpha_i,\alpha_j)$ and $V = V(Q_{ij},u_i,u_j,u_{ij},\alpha_i,\alpha_j)$.
\end{Proposition}

\subsection{Double zeroes on octahedron equations}

Some of the double zero expansions above can be written in terms of octahedron polynomials. For A2, we can rewrite the expansion~\eqref{Striangle-E} using equation~\eqref{Omegas-mu-E}.\footnote{This would also be the case for Q4, but we have not proved a closure relation for this equation, so we will not discuss its double zero expansion either.} With the cyclic symmetry and point inversion symmetry of $\Omega_1$ and $\Omega_2$, we can conclude the following.

\begin{Proposition}
	For equation ${\rm A}2$, a double zero expansion for the cube action of the triangle $2$-form in terms of $\Omega_1$, $\Omega_2$ is given by
	\begin{align*}
		S_\triang^{\vec \Theta, \vec \Xi}={}&{-}\frac{1}{2} \sum_{n = 2}^\infty \frac{ \frac{\partial^n}{\partial {Q_{ij}}^n} \cL_\crosssquare(V,u_i,u_j,u_{ij},\alpha_i,\alpha_j) \big|_{Q_{ij}=0} }{ n! \bigl( \pdv{T}{u} \frac{\mu(u_i,u_j,u_{jk},u_{ki},\alpha_i,\alpha_j,\alpha_k) }{\gamma(\alpha_i,\alpha_j,\alpha_k)} \bigr)^n} \biggl( \pdv{\Omega_1}{u_{k}}\Omega_2 - \pdv{\Omega_2}{u_{k}} \Omega_1\biggr)^n \\
 &{} {-} \reverse + \cyclic ,
	\end{align*}
 where $V = V(Q_{ij},u_i,u_j,u_{ij},\alpha_i,\alpha_j)$.
\end{Proposition}
For the rest of the ABS list, we can rewrite the expansion~\eqref{Striangle-E} using equation~\eqref{Omegas-explicit-E} and can conclude the following.

\begin{Proposition}
	For equations ${\rm Q}1$, ${\rm Q}2$, ${\rm Q}3$, ${\rm H}1$, ${\rm H}2$, ${\rm H}3$ and ${\rm A}1$, a double zero expansion for the cube action of the triangle $2$-form in terms of $\Omega_1$, $\Omega_2$ is given by
	\begin{align*}
		S_\triang^{\vec \Theta, \vec \Xi}={} &{-}\frac{1}{2} \sum_{n = 2}^\infty \frac{\frac{\partial^n}{\partial {Q_{ij}}^n} \cL_\crosssquare(V,u_i,u_j,u_{ij},\alpha_i,\alpha_j) \big|_{Q_{ij}=0} }{n! \bigl( \pdv{T}{u} \frac{1}{\gamma(\alpha_i,\alpha_j,\alpha_k)} g_1 \bigl( \pdv{\Omega_1}{u_{ij}} - \pdv{\Omega_1}{u_{k}} \bigr) \bigr)^n} \biggl( \pdv{\Omega_1}{u_{k}}\Omega_2 - \pdv{\Omega_2}{u_{k}} \Omega_1 \biggr)^n \\
 &{}{-} \reverse + \cyclic ,
	\end{align*}
 where $V = V(Q_{ij},u_i,u_j,u_{ij},\alpha_i,\alpha_j)$.
\end{Proposition}
These propositions give a possible variational interpretation to the octahedron equations. The~double zero expansions imply that the equations $\Omega_1 = 0$ and $\Omega_2 = 0$ are sufficient conditions for criticality, hence they imply the corner equations.

\subsection{Example: Double zero expansions for H1}

In this subsection, we show that the general construction described above leads to succinct double zero expansions for the actions over the cube of each of the discrete 2-forms associated with H1. For this example, we will verify by explicit computation that the zeroth and first order terms vanish in the Taylor expansions.

{\bf Cross-square Lagrangian.} Solving the quad equation $Q_{ij} = 0$ for $u$, with $Q_{ij}$ given by equation~\eqref{H1-Q}, we find a change of variables expressing $u$ in terms of the quad polynomial $Q_{ij}$:
\begin{equation*}
 u = \frac{\alpha_i -\alpha_j + Q_{ij}}{u_i - u_j} + u_{ij} .
\end{equation*}
Now we apply this to a single cross-square Lagrangian and Taylor expand in $Q_{ij}$
\begin{align*}
 \cL_\crosssquare &= u u_i + u_j u_{ij} - u u_j - u_i u_{ij} - (\alpha_i - \alpha_j) \log(u - u_{ij}) - (\alpha_i - \alpha_j) \log(u_i - u_j) , \\
 &= ( u - u_{ij} ) (u_i - u_j) - (\alpha_i - \alpha_j) \log((u - u_{ij})(u_i - u_j) ) + 2 \pi {\rm i} \Xi (\alpha_i-\alpha_j) , \notag \\
 & = \alpha_i - \alpha_j + Q_{ij} - (\alpha_i - \alpha_j) \log(\alpha_i - \alpha_j + Q_{ij}) + 2 \pi {\rm i} \Xi (\alpha_i-\alpha_j) \\
 & = (\alpha_i - \alpha_j)(1 - \log(\alpha_i - \alpha_j)) + 2 \pi {\rm i} \Xi (\alpha_i-\alpha_j) - \sum_{n=2}^\infty \frac{Q_{ij}^n}{n (\alpha_j - \alpha_i)^{n-1}} .
\end{align*}
We find that the first order term vanishes as expected, but the zeroth order term does not. It~depends only on the lattice parameters. When we consider the action around the cube, the zeroth order contributions cancel, except for terms of the form $2 \pi {\rm i} \Xi_i \alpha_i$:
\begin{align*}
 S_\crosssquare & = \inversion{\cL_\crosssquare} - \cL_\crosssquare + \cyclic \, \\
 & = \sum_{n=2}^\infty \frac{Q_{ij}^n}{n (\alpha_j - \alpha_i)^{n-1}} - \reverse + \cyclic - 2 \pi {\rm i} ( \Xi_i \alpha_i + \Xi_j \alpha_j + \Xi_k \alpha_k) .
\end{align*}
Hence,
\begin{equation}
 \label{H1-double0-Q}
S_\crosssquare^{\vec \Xi} = \sum_{n=2}^\infty \frac{Q_{ij}^n}{n (\alpha_j - \alpha_i)^{n-1}} - \reverse + \cyclic .
\end{equation}
Equation~\eqref{H1-double0-Q} is a double zero expansion for the action around the cube of the cross square 2-form associated with H1 in terms of $Q_{ij}$, $Q_{jk}$, $Q_{ki}$, \smash{$\inversion Q_{ij}$}, \smash{$\inversion Q_{jk}$}, \smash{$\inversion Q_{ki}$}.

{\bf Cross Lagrangian.}
In order to derive a double zero expansion for the cross 2-form associated with H1, we use its tetrahedron expression
\[ T = (\alpha_i-\alpha_j)(u u_{ij} + u_{jk}u_{ki})
+ (\alpha_j-\alpha_k)(u u_{jk} + u_{ij} u_{ki})
+ (\alpha_k-\alpha_i)(u u_{ki} + u_{ij} u_{jk} ) \]
to find a variable transformation that eliminates $u$ in favour of $T$:
\[
 u = \frac{T - ( (\alpha_i-\alpha_j) u_{jk} u_{ki} + \cyclic ) }{( (\alpha_i - \alpha_j) u_{ij} + \cyclic ) } .
\]
From this, we obtain
\begin{equation}
 \label{H1-T-for-u}
 u - u_{ij} = \frac{T + (\alpha_i-\alpha_j) (u_{ij} - u_{jk})(u_{ki} - u_{ij}) }{( (\alpha_i - \alpha_j) u_{ij} + \cyclic ) } .
\end{equation}
We have an analogous transformation for $u_{ijk}$ in terms of \smash{$\inversion T$}. Now we apply this to the action around the cube of the cross 2-form and Taylor expand in $T$ and \smash{$\inversion T$}.
Using the point inversion symmetry, we find
\begin{align*}
\begin{split}
 S_\cross
 & = -\cL_\cross -\reverse + \cyclic \\
 &= -(\alpha_i - \alpha_j) \log(u_i - u_j) + (\alpha_i - \alpha_j) \log(u - u_{ij}) -\reverse + \cyclic \\
 &= (\alpha_i - \alpha_j) \log(u_{jk} - u_{ki}) + (\alpha_i - \alpha_j) \log(u - u_{ij}) -\reverse + \cyclic \\
 &= (\alpha_i - \alpha_j) \log ((u - u_{ij})(u_{jk} - u_{ki})) - \reverse + 2 \pi {\rm i} \Xi_{ij} (\alpha_i-\alpha_j) + \cyclic .
 \end{split}
\end{align*}
Now we substitute $u$ using equation~\eqref{H1-T-for-u},
\begin{align*}
 S_\cross &= (\alpha_i - \alpha_j) \log \biggl( \frac{T(u_{jk} - u_{ki}) + (\alpha_i-\alpha_j) (u_{ij} - u_{jk})(u_{ki} - u_{ij})(u_{jk} - u_{ki}) }{\bigl( (\alpha_i - \alpha_j) u_{ij} + \cyclic \bigr) } \biggr) \\
 &\quad{}- \reverse + 2 \pi {\rm i} \Xi_{ij} (\alpha_i-\alpha_j) + \cyclic .
\end{align*}
Then we observe that
\begin{align*}
& (\alpha_i - \alpha_j) \log \bigl( (u_{ij} - u_{jk})(u_{ki} - u_{ij})(u_{jk} - u_{ki}) \bigr) + \cyclic= 0 , \\
& (\alpha_i - \alpha_j) \log \bigl( (\alpha_i - \alpha_j) u_{ij} + \cyclic \bigr) + \cyclic= 0 ,
\end{align*}
because the logarithms are invariant under cyclic permutations of $i$, $j$, $k$, so we find
\begin{align*}
 S_\cross &= (\alpha_i - \alpha_j) \log \biggl( \frac{T}{(u_{ij} - u_{jk})(u_{ki} - u_{ij})} + (\alpha_i-\alpha_j) \biggr) - \reverse + \cyclic \\
 &\quad{}- 2 \pi {\rm i} ( \Xi_i \alpha_i + \Xi_j \alpha_j + \Xi_k \alpha_k) \\
 & = -\sum_{n=2}^\infty \frac{T^n}{n (\alpha_i - \alpha_j)^{n-1} (u_{ij} - u_{jk})^n (u_{ij} - u_{ki})^n} - \reverse + \cyclic \\
 &\quad{}- 2 \pi {\rm i} ( \Xi_i \alpha_i + \Xi_j \alpha_j + \Xi_k \alpha_k) .
\end{align*}
Hence,
\begin{equation}
 \label{H1-double0-T}
 S_\cross^{\vec \Xi} = -\sum_{n=2}^\infty \frac{T^n}{n (\alpha_i - \alpha_j)^{n-1} (u_{ij} - u_{jk})^n (u_{ij} - u_{ki})^n} - \reverse + \cyclic .
\end{equation}
Note that the zeroth and first order terms vanish. Hence, equation~\eqref{H1-double0-T} is a double zero expansion in terms of $T$ and \smash{$\inversion T$}.

{\bf Triangle Lagrangian.}
In order to derive a double zero expansion for the action on the elementary cube of the triangle 2-form, we note the following relation for H1:
\begin{equation*}
 T = 0 \ \implies \ Q_{ij} = \frac{(\alpha_i - \alpha_j) E_{ij}}{(\alpha_i (u_{ij} - u_{ki}) + \cyclic)} .
\end{equation*}
There holds \smash{$T,\inversion T = 0 \implies S_\cross = 0$}, so we deduce form equation~\eqref{Striang-crosssquarecross} that
\[ T,\inversion T = 0 \ \implies\ S_\triang
= \frac{1}{2} S_\crosssquare . \]
Since $S_\triang(u_i,u_j,u_k,u_{ij},u_{ji},u_{ki})$ does not depend on $u$ or $u_{ijk}$, we can assume that $T = 0$ and~\smash{$\inversion T = 0$}, without affecting the action.
Thus, we can write the cube action for the triangle 2-form~as
\begin{align}
 S_\triang^{\vec \Xi} &= \frac{1}{2}\sum_{n=2}^\infty \frac{1}{n (\alpha_j - \alpha_i)^{n-1}} \biggl( \frac{(\alpha_i - \alpha_j) E_{ij}}{(\alpha_i (u_{ij} - u_{ki}) + \cyclic)} \biggr)^n - \reverse + \cyclic \, \notag\\
 & = \frac{1}{2}\sum_{n=2}^\infty \frac{(\alpha_j - \alpha_i)(-1)^n}{n (\alpha_i (u_{ij} - u_{ki}) + \cyclic)^n} E_{ij}^n - \reverse + \cyclic .
 \label{H1-double0-E}
\end{align}
Equation~\eqref{H1-double0-E} is a double zero expansion for the triangle 2-form in terms of $E_i$, $E_j$, $E_k$, \smash{$\inversion E_i$}, \smash{$\inversion E_j$}, \smash{$\inversion E_k$}.

We can use the identities~\eqref{Omegas-explicit-E}, where $g_1 = 1$ for H1, to rewrite this in terms of $\Omega_1$ and $\Omega_2$:
\begin{align}
 S_\triang^{\vec \Xi} & = \frac{1}{2}\sum_{n=2}^\infty \frac{(\alpha_j - \alpha_i)(-1)^n}{n (\alpha_i (u_{ij} - u_{ki}) + \cyclic)^n} \Biggl( \frac{\pdv{\Omega_1}{u_k} \Omega_2 - \pdv{\Omega_2}{u_k} \Omega_1}{\pdv{\Omega_1}{u_{ij}} - \pdv{\Omega_1}{u_k}} \Biggr)^n - \reverse + \cyclic .
 \label{H1-double0-O}
\end{align}
Equation~\eqref{H1-double0-O} is a double zero expansion for the triangle 2-form in terms of the octahedron polynomials.

{\bf Trident Lagrangian.}
To derive the double zero expansion for action of the trident 2-form, we consider equation~\eqref{T-cyclic-Qs}, which shows that cyclic combinations of quad equations lead to the tetrahedron equation. For H1, this can be written explicitly as
\begin{equation*}
T = - (u - u_{jk}) (u - u_{ki}) Q_{ij} - (u - u_{ki}) (u - u_{ij}) Q_{jk} - (u - u_{ij}) (u - u_{jk}) Q_{ki} .
\end{equation*}
Using equation~\eqref{Strileg-crosscrosssquare}, we can write the action on an elementary cube of the trident 2-form as
\begin{align}
S_\trileg^{\vec \Xi}
 ={}& \frac{1}{2}\sum_{n=2}^\infty \biggl(\frac{Q_{ij}^n}{n (\alpha_j - \alpha_i)^{n-1}} - \frac{T^n}{n (\alpha_i - \alpha_j)^{n-1} (u_{ij} - u_{jk})^n (u_{ij} - u_{ki})^n} \biggr) - \reverse + \cyclic \nonumber \notag\\
 ={}& \frac{1}{2}\sum_{n=2}^\infty \biggl(\frac{Q_{ij}^n}{n (\alpha_j - \alpha_i)^{n-1}} - \frac{(- Q_{ij} ( u - u_{jk} ) ( u - u_{ki} ) + \cyclic)^n}{n (\alpha_i - \alpha_j)^{n-1} (u_{ij} - u_{jk})^n (u_{ij} - u_{ki})^n} \biggr) \nonumber\\
 & - \reverse + \cyclic .\!
 \label{H1-dobule0-QQQ}
\end{align}
Equation~\eqref{H1-dobule0-QQQ} is a double zero expansion for the cube action of the trident 2-form in terms of the quad polynomials $Q_{ij}$, $Q_{jk}$, $Q_{ki}$, \smash{$\inversion Q_{ij}$}, \smash{$\inversion Q_{jk}$}, \smash{$\inversion Q_{ki}$}.

\section{Conclusion}

In this work, we compared the Lagrangian multiforms $\cL_\trileg$, $\cL_\cross$ and $\cL_\crosssquare$ on four-point stencils to the~three-point one $\cL_\triang$ that has been favoured in the previous literature on Lagrangian multiforms for the ABS equations.

In Part I~\cite{richardson2025discrete1}, we focused on $\cL_\trileg$: we introduced integer fields into the action to deal with the branch cuts of the logarithm and dilogarithm functions that occur in $\cL_\trileg$, and showed that these integer fields are essential to obtain the closure property of Lagrangian multiforms.

In Part II, we showed that the same construction applies to $\cL_\cross$, $\cL_\crosssquare$, and $\cL_\triang$. We established that the closure property holds for each of them, in the sense that the action over the cube vanishes \big(modulo $4 \pi^2$\big) on solutions of the corresponding corner equations.
We studied the relations between the respective systems of corner equations using the relations between the Lagrangian 2-forms. We showed that the quad equations are equivalent to the combined system of tetrahedron and octahedron equations. In addition, we formulated the double-zero property, which has recently seen a lot of emphasis in the continuous and semi-discrete settings, on the discrete level. This gives a variational interpretation of the octahedron equations.

This work shows once again that Lagrangian multiform theory is a powerful variational principle. As well as being an attribute of integrability, it can provide a solution to the inverse problem of the calculus of variations for equations that do not admit a Lagrangian in the traditional sense.

The Lagrangian multiforms presented here are specific to the ABS equations, relying in~particular on their three-leg forms. Lagrangian multiforms exists for other integrable quad equations, such as the discrete Gel'fand--Dikii equations~\cite{lobb2010lagrangian}, but these do no have a three-leg structure. It remains to be seen if such examples also admit different Lagrangian multiforms producing distinct but related corner equations.

The double-zero property has been observed in Lagrangian multiforms for non-ABS lattice equations~\cite{nijhoff2024lagrangianb, nijhoff2024lagrangian}. Since the corner equations imply that the gradient of the action vanishes, we~expect the double-zero property to hold for all discrete Lagrangian multiforms that satisfy the closure property. A~careful study of the closure property for non-ABS examples, generalising the results of Part~I, is left for future work.

\subsection*{Acknowledgements}

The authors are grateful to Professor Frank Nijhoff and Dr Vincent Caudrelier, for their continued support over many years, and to the organisers and supporters of the BIRS-IASM workshop on Lagrangian Multiform Theory and Pluri-Lagrangian Systems in October 2023, where this work was started. We thank the anonymous referees for their constructive criticism on the initial version of this paper, which inspired much of the work presented in Part~I.

JR acknowledges funding from Engineering and Physical Sciences Research Council DTP, Crowther Endowment and School of Mathematics at the University of Leeds. MV is supported by the Engineering and Physical Sciences Research Council [EP/Y006712/1].

\pdfbookmark[1]{References}{ref}
\LastPageEnding

\end{document}